%% file: dili_main.tex
\documentclass[3p,11pt,letterpaper]{elsarticle}

\usepackage{amsbsy}
\usepackage{amssymb}
\usepackage{amsthm}
\usepackage{amsmath}
\usepackage{bm}
\usepackage{graphicx}
\usepackage{mathtools}
\usepackage{multirow}
\usepackage{rotating}
\usepackage{subfigure}
\usepackage{url}
\usepackage{color}
\usepackage{natbib}
\usepackage{todonotes}
\usepackage{ulem}
\usepackage{algorithm} 
\usepackage{algpseudocode}
\usepackage[T1]{fontenc}
\usepackage[colorlinks=true,citecolor=blue,linkcolor=blue,urlcolor=blue]{hyperref}
\usepackage{geometry}

\newcommand{\added}[1]{ { \color{blue}#1 } }
\renewcommand{\added}[1]{ #1 } 

\include{new_commands}

\journal{Journal of Computational Physics}

\begin{document}


\input{cover_r1}

\input{introduction_r1}

\input{background_r1}


\section{Likelihood-informed MCMC}
\label{sec:lkd_informed}

\added{
In this section, we first introduce a general class of operator-weighted MCMC proposals that guarantee dimension-independent sampling; within this class, there is considerable flexibility to capture the structure of the posterior by choosing the associated operators appropriately.
We then take a detour to discuss the notion of a likelihood-informed subspace (LIS), which describes the parameter subspace where the posterior distribution departs most strongly from the prior.
Based on the LIS, we derive a low-rank posterior covariance approximation scheme and then use both of these building blocks to guide the inhomogeneous discretization of a preconditioned Langevin SDE.
This discretization naturally leads to the construction of various operator-weighted proposals that fall within our general class.
We provide several specific examples of such proposals and then describe an adaptive posterior sampling framework that combines LIS construction, low-rank posterior covariance approximation, and automatic proposal adaptation.  
}

\input{opt_weighted_r1}
\input{lis_r1}

\input{li_opt_r1}

\input{algo_opt_r1}


\input{elliptic_r1}

\input{conditioned_diffusion_r1}

\input{conclusion_r1}


\section*{Acknowledgments}

T.\ Cui and Y.~M.\ Marzouk acknowledge financial support from the US Department of Energy, Office of Advanced Scientific Computing Research (ASCR), under grant number DE-SC0009297, as part of the DiaMonD Multifaceted Mathematics Integrated Capability Center.
K.J.H.~Law was supported by the King Abdullah University of Science and Technology (KAUST) SRI UQ Center, and an Oak Ridge National Laboratory Directed Research and Development Strategic Hire grant.



\input{dili_ref}
\input{appendix_r1}

\end{document}

%% file: new_commands.tex
\newtheorem{theorem}{Theorem}[section]
\newtheorem{assumption}[theorem]{Assumption}
\newtheorem{corollary}[theorem]{Corollary}
\newtheorem{definition}[theorem]{Definition}
\newtheorem{proposal}{Proposal}[section]

\newdefinition{remark}[theorem]{Remark}

\newcommand{\dt}{\triangle \tau}
\newcommand{\im}{{\rm Im}}

\newcommand{\lkd}{\mathcal{L}}
\newcommand{\potential}{\eta}

\newcommand{\expect}{\mathbb{E}}
\newcommand{\forward}{F}
\newcommand{\linear}{\mathrm J}

\newcommand{\gradientu}{\nabla_u\potential(u;y)}
\newcommand{\gradientv}{\nabla_v\potential(v;y)}
\newcommand{\gradientw}{\nabla_w\potential(w;y)}

\newcommand{\hessian}{\mathrm H}
\newcommand{\precondition}{\mathrm P}

\newcommand{\normal}{\mathcal{N}}
\newcommand{\hilbert}{\mathcal{H}}
\newcommand{\banach}{\mathcal{X}}
\newcommand{\dataspace}{\mathbb{R}^{N_y}}

\newcommand{\pocov}{\Gamma_{\mathrm{pos}}}
\newcommand{\prcov}{\Gamma_{\mathrm{pr}}}
\newcommand{\obscov}{\Gamma_{\mathrm{obs}}}

\newcommand{\prmean}{m_0}

\newcommand{\rand}{\xi}
\newcommand{\eigentrunc}{\varrho}

\newcommand{\N}{\mathbb{N}}
\newcommand{\R}{\mathbb{R}}

\newcommand{\cA}{\mathrm A}
\newcommand{\cB}{\mathrm B}
\newcommand{\cG}{\mathrm G}
\newcommand{\cD}{\mathrm D}
\newcommand{\cI}{\mathrm I}

\newcommand{\cW}{\mathrm W}
\newcommand{\basis}{\mathrm \Psi}

\newcommand{\cM}{\mathrm M}
\newcommand{\cS}{\mathrm S}

\newcommand{\cC}{\mathrm C}
\newcommand{\cL}{\mathrm L}
\newcommand{\cT}{\mathrm T}
\newcommand{\cQ}{\mathrm Q}
\newcommand{\cR}{\mathrm R}

\newcommand{\ray}{\mathcal{R}}
\newcommand{\white}{d\mathcal{W}}
\newcommand{\sub}{\mathcal{V}}

\newcommand{\half}{{1/2}}

%% file: cover_r1.tex

\begin{frontmatter}
  
  \title{Dimension-independent likelihood-informed MCMC}
    
  \author[mit]{Tiangang Cui}
  \ead{tcui@mit.edu}

  \author[oak]{Kody J.~H.~Law}
  \ead{lawkj@ornl.gov}
   
  \author[mit]{Youssef M.~Marzouk}
  \ead{ymarz@mit.edu}

  \address[mit]{Massachusetts Institute of Technology, Cambridge, MA 02139, USA}
  \address[oak]{Computer Science and Mathematics Division, Oak Ridge National Laboratory, Oak Ridge, TN 37934, USA}

  \begin{abstract}
    Many Bayesian inference problems require exploring the posterior
    distribution of high-dimensional parameters that represent the
    discretization of an underlying function.
    This work introduces a family of Markov chain Monte Carlo (MCMC)
    samplers that can adapt to the particular structure of a posterior
    distribution over functions.
    Two distinct lines of research intersect in the methods developed
    here.
    First, we introduce a general class of operator-weighted
    proposal distributions that are well defined on function space,
    such that the performance of the resulting MCMC samplers is
    independent of the discretization of the function.
    Second, by exploiting local Hessian information and any associated
    low-dimensional structure in the change from prior to posterior
    distributions, we develop an inhomogeneous discretization scheme
    for the Langevin stochastic differential equation that yields
    operator-weighted proposals adapted to the non-Gaussian structure
    of the posterior.
    The resulting \textit{dimension-independent} and
    \textit{likelihood-informed} (DILI) MCMC samplers may be useful
    for a large class of high-dimensional problems where the target
    probability measure has a density with respect to a Gaussian
    reference measure.
    Two nonlinear inverse problems are used to demonstrate the
    efficiency of these DILI samplers: an elliptic PDE coefficient
    inverse problem and path reconstruction in a conditioned diffusion.

  \end{abstract}

  \begin{keyword}
	Markov chain Monte Carlo \sep likelihood-informed subspace \sep infinite-dimensional inverse problems \sep Langevin SDE \sep conditioned diffusion
  \end{keyword}
  
\end{frontmatter}

%% file: introduction_r1.tex

\section{Introduction}
\label{sec:intro}
Many Bayesian inference problems require sampling a posterior distribution of high-dimensional parameters that represent the discretization of some underlying \textit{function}. Examples include inverse problems governed by partial differential equations (PDEs) \cite{Tarantola_2004, Kaipio_2005, Stuart_2010} and path reconstructions of stochastic differential equations \cite{BPRF_2006, HSV_2011}. Yet sampling the posterior with standard Markov chain Monte Carlo (MCMC) algorithms can become arbitrarily slow as the representation of the unknown function is refined; the convergence rates of such algorithms typically degrade with increasing parameter dimension \cite{RGG_1997, Roberts_2001, Mattingly_2012, PST_2012}. Here we present a \textit{dimension-independent} and \textit{likelihood-informed} MCMC sampling framework that enables efficient exploration of posterior distributions over functions.
Our framework builds on recent work in dimension-independent MCMC sampling by identifying changes in the local curvature of the log-posterior, relative to the log-prior, and using this information to design more effective proposal distributions. We cast these proposals within a general class of operator-weighted proposals that guarantee dimension independence while allowing flexibility in the design of structure-exploiting MCMC algorithms.

Previous work \cite{BRSV_2008, CRSW_2012} has shown that it is possible to overcome degeneration of the convergence rates of standard MCMC algorithms by formulating them in function space. In particular, \cite{BRSV_2008, CRSW_2012} consider a posterior that is absolutely continuous with respect to a Gaussian reference measure, typically chosen to be the prior. By deriving proposals that yield valid Markov transition kernels on a function space, these efforts yield algorithms with performance that is \textit{dimension independent}, i.e., robust under mesh refinement.
Recent work has begun to integrate this function space perspective with other state-of-the-art sampling approaches such as sequential Monte Carlo~\cite{Chopin_2002, SMC_2006} and particle MCMC~\cite{ADH_2010} (which were originally designed for finite-dimensional parameters) along with multi-level Monte Carlo~\cite{Giles_2008}. 
Examples include \cite{BKJ_2013}, which uses function space proposals inspired by those in \cite{BRSV_2008, CRSW_2012} as building blocks within a sequential Monte Carlo scheme, and \cite{HSS_2013, KST_2013}, which employ function space algorithms within a multi-level MCMC scheme for variance reduction. An important takeaway from these efforts is that although proposals that are well-defined on function space provide mixing rates that are independent of discretization, they do not guarantee good mixing in absolute terms; additional analysis and exploitation of the structure of the posterior is essential.

To improve MCMC mixing in finite-dimensional settings, another important line of research has focused on using information about the local geometry of the posterior distribution to scale and guide proposals.
Examples include the stochastic Newton method \cite{Martin_2012}, which uses low-rank approximations of the Hessian of the log-posterior to build Gaussian proposals, and Riemannian manifold MCMC \cite{Girolami_2011}, which uses the Fisher information matrix (and the prior precision matrix) to construct a local metric for both Langevin proposals and Hamiltonian dynamics. Implicit in the low-rank approximations of \cite{Martin_2012} is the idea that the Hessian of the log-likelihood can be compared with the prior covariance to identify directions in parameter space along which the posterior distribution differs most strongly from the prior.  This idea is formalized for linear inverse problems in~\cite{Linear_Redu_2014}; at its heart lies the solution of a generalized eigenproblem involving both operators. In linear-Gaussian problems, a transformation of the resulting dominant eigendirections reveals the directions of greatest \textit{variance reduction}, relative to the prior. Solutions of the eigenproblem can also be used to construct optimal approximations of the posterior covariance, in the form of low-rank negative semidefinite \textit{updates} of the prior covariance \cite{Linear_Redu_2014}. The form of these approximations is crucial, as it suggests the following tie to dimension-independent algorithms: The proposals used in \cite{BRSV_2008, CRSW_2012} rely on sampling from the Gaussian prior measure on an infinite-dimensional space. It is natural to alter these proposals only in a \textit{finite} number of dimensions---where the posterior differs most strongly from the prior---to achieve better mixing while preserving dimension independence.

A first attempt in this direction is the operator-weighted proposal of \cite{Proposal_Law_2012}, which employs a relaxed version of the log-likelihood Hessian, evaluated at a single point near the posterior mode, to rescale the preconditioned Crank-Nicolson proposal of \cite{CRSW_2012}. Of course, in inverse problems with nonlinear forward operators or non-Gaussian noise, the log-likelihood Hessian varies over the parameter space. Nonetheless, if local Hessians have low rank and dominant eigenspaces that are somewhat aligned, changes from the prior to the posterior will be confined to a relatively low-dimensional subspace of the parameter space. We will call this subspace the ``likelihood-informed'' subspace (LIS). \cite{Cui_LIS_2014} introduced the notion of the LIS and constructed it from the dominant eigenspace of the posterior expectation of a preconditioned log-likelihood Hessian. There, the LIS was used to \textit{approximate} the posterior distribution in finite-dimensional nonlinear inverse problems, by replacing the posterior in complementary directions with the prior. In inverse problems with limited observations and smoothing forward operators, the LIS is relatively low dimensional and such approximations can be quite accurate. 

We will use the notion of an LIS in the present work, but extend it to the infinite-dimensional setting and seek exact sampling; in this context, the LIS can be seen as means of convergence acceleration that captures key features of the posterior. The goal of this paper is to design MCMC algorithms that are dimension independent, that sample from the exact posterior distribution, and that exploit a posterior structure wherein departures from the prior---including non-Gaussianity---are concentrated on a finite number of directions. To realize these goals, we introduce a general class of proposal distributions that simultaneously guarantee dimension-independent MCMC sampling and allow the essential structure of the posterior (i.e., departures from the Gaussian prior) to be captured in a flexible manner, via both local gradient information and a set of global covariance-like (i.e., bounded and self-adjoint) operators. We then discuss particular instances of this class for which the operators are derived from the LIS. The proposals themselves can be obtained from an underlying Langevin SDE by applying an inhomogeneous time discretization scheme on the LIS and its complement. We then incorporate these proposals into MCMC either as Metropolis or Metropolis-within-Gibbs updates.

The rest of this paper is organized as follows. Section~\ref{sec:bkgd} reviews some theoretical background on the infinite-dimensional Bayesian formulation of inverse problems and function-space MCMC methods, as well as previous work aimed at improving the efficiency of these samplers.
Section~\ref{sec:lkd_informed} introduces a general class of operator-weighted proposals that enables the design of dimension-independent and likelihood-informed (DILI) samplers. We then discuss the construction of a global LIS and describe several different DILI MCMC samplers based on the LIS. We also recall a few non-DILI algorithms that will subsequently be used for benchmarking.
In Section~\ref{sec:elliptic}, an elliptic PDE inverse problem is used to evaluate the performance of the various samplers introduced in Section \ref{sec:lkd_informed}. In Section~\ref{sec:cd}, we apply the same samplers to the problem of path reconstruction in a conditioned diffusion and highlight key differences in performance. Section \ref{sec:conc} offers some brief discussion and concluding remarks.

%% file: background_r1.tex

\section{Theoretical background}
\label{sec:bkgd}

In this section we recall the Bayesian formulation of inverse problems in the function space setting and provide some background on the relevant MCMC methods. We then review preconditioned Crank-Nicolson (pCN) proposals, which yield dimension-independent MCMC algorithms for posterior distributions over function space equipped with an appropriate Gaussian prior. We close this section with a discussion of how local geometry information can be used to inform MCMC proposals.

\subsection{Bayesian inference framework}
\label{sec:bayes}

Suppose that the parameter of interest is some function $u \in \hilbert$, where $\hilbert$ is a separable Hilbert space.
Suppose further that the prior distribution $\mu_0$ is such that $\mu_0(\hilbert)=1$. We assume that the observation $y$ is in $\dataspace$.
Let $\lkd(u;y)$ denote the likelihood function of $y$ given $u$, which is assumed to be bounded away from zero $\mu_0$-almost surely.
Let $\mu_y$ denote the posterior probability measure of the unknown $u$ conditioned on the observation $y$.
The posterior distribution on infinitesimal volume elements $du$ in $\hilbert$ is given by 
\begin{equation}
\mu_y(du) \propto \lkd(u;y) \mu_0(du).
\label{eq:posterior}
\end{equation}
Let $\forward : \banach \rightarrow \dataspace$ denote the forward operator, defined on 
some Banach space $\banach \subseteq \hilbert$ with $\mu_0(\banach)=1$, 
and assume the noise on the observation is additive and Gaussian with zero mean.
Then the observation model has the form
\[
y = \forward(u) + e, \quad e \sim \normal(0,\obscov).
\]
We define the data-misfit functional 
\begin{equation}
\potential(u; y) = \frac{1}{2}\left( \forward(u) - y \right)^\top \obscov^{-1} \left( \forward(u) - y \right).
\label{eq:loglike}
\end{equation}
The likelihood function can then be written as
\begin{equation}
\lkd(u;y) \propto \exp(-\potential(u; y)).
\label{eq:like}\end{equation}

The inner products on the data space $\dataspace$ and the parameter space $\hilbert$ are denoted $\langle \cdot \, , \cdot \rangle$ and $\langle \cdot \, , \cdot \rangle_\hilbert$, respectively, with associated norms denoted by $\| \cdot \|$ and $\| \cdot \|_\hilbert$.
For brevity, where misinterpretation is not possible, we will drop the subscripts $\hilbert$.
Throughout this paper we adopt the following assumptions about the forward operator:
\begin{assumption}
\label{assum1}
The forward operator $\forward: \banach \rightarrow \dataspace$ satisfies the following:
\begin{enumerate}
\item For all $\epsilon_1 > 0$, there exists a constant $K_1(\epsilon_1) > 0$ such that, for all $u \in \banach$, 
\[
\| \forward(u) \| \leq \exp\left( K_1(\epsilon_1) + \epsilon_1\|u\|_\hilbert^2 \right) .
\]
\item For all $\epsilon_2 > 0$ there exists a constant $K_2(\epsilon_2) > 0$ such that, for all $u_1, u_2 \in \banach$ 
with $\|u_1\|_\hilbert < \epsilon_2$ and $\|u_2\|_\hilbert < \epsilon_2$,
\[
\| \forward(u_1) - \forward(u_2) \| \leq K_2(\epsilon_2) \|u_1 - u_2\|_\hilbert.
\]
\item For all $u \in X$, there exists a bounded linear operator $\linear: \banach \rightarrow \dataspace$ such that for all 
$\delta u \in \banach$
\[
\lim_{\delta u \rightarrow 0} \frac{\|\forward(u+\delta u) - \forward(u) - \linear(u) \delta u\|}{\|\delta u\|_\hilbert} = 0.
\]
\end{enumerate}
\end{assumption}
Given observations $y$ such that $\|y\| < \infty$ and a forward operator that satisfies the first two conditions in Assumption \ref{assum1}, \cite{Stuart_2010} shows that the resulting data-misfit function is sufficiently bounded and locally Lipschitz, and that the posterior measure is dominated by the prior measure.
The third condition in Assumption \ref{assum1} states that the forward model is first order Fr\'{e}chet differentiable, and hence the Gauss-Newton approximation of the Hessian of the data-misfit functional is bounded.

We assume that the prior measure is Gaussian with mean $m_0 \in \hilbert$ and a covariance operator $\prcov$ on $\hilbert$ that is self-adjoint, positive definite, and trace-class, so that the prior provides a full probability measure on $\hilbert$. 
The covariance operator $\prcov$ also defines the inner product $\langle \cdot , \,  \cdot \rangle_{\prcov} = \langle \prcov^{-1/2} \, \cdot \, , \prcov^{-1/2} \, \cdot \rangle$ and the associated norm $\|\cdot\|_{\prcov} $ on the Hilbert space $\hilbert$.
It is convenient to define the covariance operator through its inverse $\prcov^{-1}$, referred to as the prior precision and often chosen to be a Laplace-like differential operator \cite{Stuart_2010}. 

We note that the posterior distribution $\mu_y$ does not admit a density with respect to Lebesgue measure.
Instead of a ``posterior density,'' we can define the limiting ratio of the probabilities of shrinking balls at two points $u_1$ and $u_2$ in $\im(\prcov^\half)$ by the exponential of the difference of the Onsager-Machlup functional (OMF),
\begin{equation}
\potential(u;y) + \|u - m_0\|_{\prcov}^2,
\label{eq:omf}
\end{equation}
evaluated at $u_1$ and $u_2$.
Then the maximizer of the probability over vanishing balls, the {\it maximum a posteriori} (MAP) estimator, is the minimizer of this functional,\footnote{The functional is extended to $\hilbert$ by prescribing it to take the value $+\infty$ outside of $\im(\prcov^\half)$.} 
as shown in \cite{DLSV_2013} for problems like those considered here.

\added{Here we assume that the data $y$ are finite dimensional. 
For suitable combinations of forward model and observational noise, this setting can be generalized to problems with infinite-dimensional data sets \cite{Stuart_2010, DLSV_2013}. 
}

\subsection{MCMC on function space}
\label{sec:mcmc}

MCMC methods, in particular the Metropolis-Hastings (MH) algorithm \cite{Metropolis, Hastings}, can be employed to sample the posterior distribution $\mu_y$ over $\hilbert$.
The MH algorithm defines a Markov chain by prescribing an accept-reject move based on a proposal distribution $q$. One step of the MH algorithm is given as follows:
\begin{definition}
[Metropolis-Hastings kernel] Given the current state $u_k$, a candidate state $u^{\prime}$ is drawn from the proposal distribution $q(u_k,\cdot)$.
Define the pair of measures
\begin{equation}
\begin{array}{rll}
\label{eq:nus}
\nu(du,du^{\prime}) &=& q(u,du^{\prime}) \mu(du) \\
\nu^\bot(du,du^{\prime}) &=& q(u^{\prime},du) \mu(du^{\prime}).
\end{array}
\end{equation}
Then the next state of the Markov chain is set to $u_{k+1} = u^{\prime}$ with probability
\begin{equation} 
\alpha(u_k,u^{\prime}) = {\rm min} 
\left \{1, 
\frac{d\nu^\bot}{d\nu}(u_k,u^{\prime})
 \right \} ; 
\label{acceptance} 
\end{equation}
otherwise, with probability $1-\alpha(u_k,u^{\prime})$, it is set to $u_{k+1}=u_k$.
\end{definition}

The MH algorithm requires the absolute continuity condition $\nu^\bot \ll \nu$ to define a valid acceptance probability and hence a valid transition kernel \cite{Tierney_1998}. We will refer to a MH algorithm as \textbf{well-defined} if this absolute continuity condition is satisfied. Most of the MCMC proposals used in current literature---e.g., vanilla random walks, Metropolis adjusted Langevin (MALA) \cite{Besag_1994}, simplified manifold MALA \cite{Girolami_2011}---violate this condition in the limit $\text{dim}(\hilbert) \rightarrow \infty$. For these proposals, the MH algorithm is defined only in finite dimensions \cite{CRSW_2012} and the convergence rate degenerates under mesh refinement. Extensive investigations have been carried out to study the rate of degeneration of various proposals and optimal scaling relative to this. See \cite{RGG_1997, Roberts_1998, Roberts_2001} and references therein.

For target probability measures over function space, the sequence of papers \cite{BRSV_2008, CRSW_2012, HSV_2012, Stuart_2010} provide a viable way of constructing well-defined MH algorithms using appropriate discretizations of Langevin SDEs. 
The key step in the construction of such algorithms is to ensure the absolute continuity condition
$\nu^\bot \ll \nu$. We will revisit these constructions in the next subsection.

\subsection{Preconditioned semi-implicit proposal}
\label{sec:pcn}

Following \cite{BRSV_2008, CRSW_2012, HSV_2012}, we recall the preconditioned Crank-Nicolson (pCN) proposal. Consider the Langevin SDE 
\begin{equation}
du = - \precondition \left( \prcov^{-1} (u-m_0) + \gamma \gradientu \right)d\tau + \sqrt{2\precondition} \white(\tau),
\label{eq:langevin}
\end{equation}
where $\precondition$ is a symmetric, positive-definite, and bounded operator and $\white(\tau)$ is space-time white noise over $\hilbert \times \R_+$.
For any such $\precondition$, the SDE \eqref{eq:langevin} has invariant distribution $\mu_y$ for $\gamma=1$ and invariant distribution $\mu_0$ for $\gamma=0$; see \cite{CRSW_2012} and references therein. 

Following \cite{BRSV_2008,CRSW_2012} we let $\precondition = \prcov$.  A first order semi-implicit discretization of \eqref{eq:langevin}, parameterized by $\vartheta \in (0, 1]$, defines a family of proposal distributions:
\begin{equation}
u^{\prime} = \frac{1 - (1-\vartheta) \dt}{1 + \vartheta \dt} (u-m_0) + m_0 - \frac{\gamma \dt}{1 + \vartheta \dt} \prcov \gradientu + \frac{\sqrt{2\dt}}{1 + \vartheta \dt} \prcov^{\half} \rand,
\label{eq:implicit}
\end{equation}
where $\rand \sim \normal(0, \cI) $, $\dt$ is the time step, and $\gamma = \{0, 1\}$ is a tuning parameter to switch between the full Langevin proposal ($\gamma = 1$) and the random walk (Ornstein-Uhlenbeck) proposal ($\gamma = 0$).
Simulating the Langevin SDE thus provides a method for approximately sampling from 
$\mu_y$ when $\gamma=1$ or from $\mu_0$ when $\gamma=0$.
Under certain technical conditions, Theorem 4.1 of \cite{BRSV_2008} shows that $\vartheta = 1/2$ (a Crank-Nicolson scheme) is the unique choice that delivers well-defined MCMC algorithms when ${\rm dim}(\hilbert) \rightarrow \infty$.
Putting $\vartheta = 1/2$ and letting $a = (2 - \dt)/(2 + \dt)$, \eqref{eq:implicit} can be rewritten as
\begin{equation}
\label{eq:pcn}
u^{\prime} = a (u - m_0) + m_0 -  \gamma (1-a) \prcov \gradientu + \sqrt{1 - a^2} \prcov^{\half} \rand,
\end{equation}
where it is required that $a \in (-1,1)$. If $\gamma=0$, the invariant distribution of \eqref{eq:pcn} remains $\mu_0$.  However, if $\gamma=1$, 
the invariant distribution of the discretized system is different from the invariant distribution of \eqref{eq:langevin}. 
To remove the bias, one can employ the Metropolis correction, as discussed by \cite{Besag_1994}.

The pCN proposal \eqref{eq:pcn} has the desired dimension independence property, in that the autocorrelation of the resulting samples does not increase as the discretization of $u$ is refined.
Since the scale of the posterior distribution necessarily tightens, relative to the prior, along directions in the parameter space that are informed by the data, maintaining a reasonable acceptance probability requires that the step size used in the proposal \eqref{eq:pcn} be dictated by these likelihood-informed directions. 
As shown in \cite{Proposal_Law_2012}, the Markov chain produced by MH with the standard pCN proposal then decorrelates more quickly in the parameter subspace that is data-informed than in the subspace that is dominated by the prior. The proposed moves in the prior-dominated directions are effectively too small or conservative, resulting in poor mixing.
The underlying issue is that the scale of the proposal \eqref{eq:pcn} is uniform in all directions with respect to the norm induced by the prior covariance, 
and hence it does not adapt to the posterior scaling induced by the likelihood function.

It is therefore desirable to design proposals that adapt to this anisotropic structure of the posterior while retaining dimension independence. A first example is the operator-weighted proposal of \cite{Proposal_Law_2012}, which specifically (though without loss of generality) targets a white-noise prior, $m_0=0$ and $\prcov=\cI$:
\begin{equation}
\label{eq:mpcn}
u^{\prime} = \cA u  + \left(\cI - \cA^2\right)^{\half} \rand.
\end{equation}
Some intuition for this proposal can be had by comparing to \eqref{eq:pcn} in the case $\gamma = 0$: the new proposal replaces the scalar $a$ with a symmetric operator $\cA$. 
This operator is constructed from a regularized approximation of the Hessian of the OMF \eqref{eq:omf}
at a single posterior sample, chosen close to or at the MAP estimate. 
The dimension independence property of the proposal \eqref{eq:mpcn} can be seen simply by noticing its reversibility with respect to the prior,
$$
q(u,du^{\prime}) \mu_0(du) = q(u^{\prime},du) \mu_0(du^{\prime}),
$$
which implies
$$
\nu(du,du^{\prime}) = \exp\{\potential(u;y) - \potential(u^{\prime};y)\} \nu^\bot(du,du^{\prime})
$$
and hence the absolute continuity condition $\nu^\bot \ll \nu$ (given that Conditions (1) and (2) of Assumption \ref{assum1} also hold).
In comparison with \eqref{eq:pcn}, the proposal \eqref{eq:mpcn} produces direction-dependent step sizes and, in particular, different step sizes for parameter directions that are prior-dominated than for directions that are influenced by the likelihood---insofar as this anisotropy is captured by $\cA$. The autocorrelation time 
of the resulting MH chain can thus be substantially smaller than that of a chain using the original pCN proposal \cite{Proposal_Law_2012}.

Despite these performance improvements, this simple operator-weighted proposal has several limitations, especially when compared to state-of-the-art MCMC approaches for finite-dimensional problems. First, we note that $\cA$ only uses information from a single point in the posterior, which may not be representative of posterior structure in a non-Gaussian problem. Also, \eqref{eq:mpcn} uses no gradient information and more broadly no local information, e.g., to describe the geometry of the target measure around each $u$. Also, the particular form of $\cA$ used above does not necessarily identify directions of greatest \textit{difference} between posterior and the prior or reference measure. (This point will be explained more fully in Section~\ref{sec:scale_lis}.) A more general class of operator-weighted proposals, designed in part to address these limitations, will be presented in Section~\ref{sec:lkd_informed}.

\subsection{Explicit proposals using local geometry}

Before describing our general class of operator-weighted proposals, we briefly recall the simplified manifold MALA \cite{Girolami_2011} and stochastic Newton \cite{Martin_2012} algorithms, which are related to an explicit discretization of the Langevin SDE \eqref{eq:langevin}. Both are well-established MCMC methods that use Hessian or Hessian-related information from the posterior distribution to achieve more efficient sampling. Simplified manifold MALA uses the proposal:
\begin{equation}
u^{\prime} = u - \precondition(u) \left( \gradientu + \prcov^{-1} u \right) \dt + \sqrt{2\dt} \precondition(u)^{\half} \rand,
\label{eq:smala}
\end{equation}
where the preconditioning operator $\precondition(u)$ 
is derived from the expected Fisher information.
Stochastic Newton prescribes $\dt = 1$ and discards the factor of $\sqrt{2}$ above, which gives the proposal 
\begin{equation}
u^{\prime} = u - \precondition(u) \left( \gradientu + \prcov^{-1} u \right) + \precondition(u)^{\half} \rand,
\label{eq:sn}
\end{equation}
where the preconditioning operator $\precondition(u)$ 
is the inverse of the posterior Hessian (or a low-rank approximation thereof; see \cite{Martin_2012, Linear_Redu_2014}) at the current sample $u$.
We note that simplified manifold MALA and stochastic Newton are closely related for a likelihood function with additive Gaussian noise, as the expected Fisher information is equivalent to the Gauss-Newton approximation of the data-misfit Hessian in this case. 
Even though operations using the local preconditioner can be computationally expensive, both methods provide significant reductions in autocorrelation time per MCMC sample over standard random walks and even non-preconditioned MALA algorithms. 

In a sense, the Hessian information---or more generally, the local Riemannian metric~\cite{Girolami_2011}---used by these proposals provides two useful insights into the structure of the posterior distribution. First is an approximation of the local geometry, e.g., relative scaling along different directions in the parameter space, curvature, etc. Second, however, is insight into how the posterior distribution locally \textit{differs from a base measure}, in particular, the prior distribution. When the Hessian of the data-misfit functional is compact, departures from the prior will be limited to a finite number of directions. One can explicitly identify these directions and partition the parameter space accordingly, into a finite-dimensional subspace and its complement. Of course, this particular partition is based only on local information, so one must consider methods for globalizing it---effectively incorporating the dependence of the Hessian on $u$. Given such an extension, this partitioning suggests a useful way of designing efficient MCMC proposals in infinite dimensions: use any effective MCMC proposal on a finite-dimensional subspace, and revert to a simpler function-space sampling framework on the infinite-dimensional complement of this subspace, thus yielding a well-defined MCMC algorithm in infinite dimensions and ensuring discretization-invariant mixing. The next section will explain how to realize this idea in the context of a more general operator-weighted MCMC.

%% file: opt_weighted_r1.tex

\subsection{Operator-weighted proposals}
\label{sec:op_weighted}
We consider a general class of operator-weighted proposals from $u$ to $u'$ given by
\begin{equation}
u' =  \prcov^{\half}\cA\prcov^{-\half} (u-u_{\rm ref}) + u_{\rm ref} - \prcov^{\half} \cG \prcov^{\half} \gradientu + \prcov^{\half}\cB \rand,
\label{eq:operator_weighted}
\end{equation}
where $\cA$, $\cB$, and $\cG$ are bounded self-adjoint operators on $\im(\prcov^{-\half})$,
$\rand \sim \normal(0, \cI)$, $u_{\rm ref} = \prmean + m_{\rm ref}$ with $m_{\rm ref} \in \im(\prcov^{\half})$, and $\prcov$ is the prior covariance operator defined in Section~\ref{sec:bayes}.  
Suppose further that these operators are defined through a complete orthonormal system $(\psi_i)$ in $\im(\prcov^{-\half})$ and sequences $(a_i)$, $(b_i)$, and $(g_i)$ of positive numbers such that
\[
\cA \psi_i = a_i \psi_i, \quad \cB \psi_i = b_i \psi_i, \quad \cG \psi_i = g_i \psi_i,\quad i \in \N .
\]
Using matrix notion, let $\basis = [\psi_1, \psi_2,\ldots]$ and $\basis^{\ast} = [\psi_1,\psi_2,\ldots]^{\ast}$, where we define 
\[
\basis^{\ast}v \equiv [\langle \psi_1,v \rangle, \langle \psi_2,v \rangle, \ldots]^{\ast}.
\]
Defining the diagonal operators 
$\cD_{\cA,ij} = \delta_{ij}a_i$, $\cD_{\cB,ij} = \delta_{ij}b_i$, and $\cD_{\cG,ij} = \delta_{ij}g_i$, 
we can represent the spectral decompositions of the operators above as
$\cA = \basis \cD_\cA \basis^{\ast}$, $\cB = \basis \cD_\cB \basis^{\ast}$, and $\cG = \basis \cD_\cG \basis^{\ast}$, respectively.
The following theorem establishes conditions on the operators $\cA$, $\cB$, and $\cG$ such that the proposal \eqref{eq:operator_weighted} yields a well-defined Metropolis-Hastings algorithm for probability measures over function space.

\begin{theorem} 
\label{theo:1}
Suppose the self-adjoint operators $\cA$, $\cB$, and $\cG$ are defined by a common set of eigenfunctions $\psi_i \in \im(\prcov^{-\half})$, $ i \in \mathbb{N}$, and sequences of eigenvalues $(a_i)$, $(b_i)$, and $(g_i)$, respectively.
Suppose further that the posterior measure $\mu_y$ is equivalent to the prior measure $\mu_0 = \normal(\prmean, \prcov)$. 
The proposal \eqref{eq:operator_weighted} delivers a well-defined MCMC algorithm if the following conditions hold:
\begin{enumerate}
\item All the eigenvalues $(a_i)$, $(b_i)$, and $(g_i)$ are real and bounded, i.e., $a_i, b_i, g_i \in \R$, and there exists some $C<\infty$ such that $|a_i|, |b_i|, |g_i| < C$, $\forall i \in \N$.
\item For any index $i \in \N$ such that $b_i = 0$, we have $a_i = 1$ and $g_i = 0$.
\item For any index $i \in \N$ such that $b_i \neq 0$ we have $|b_i| \geq c $ for some $c>0$.
\item The sequences of eigenvalues $(a_i)$ and $(b_i)$ satisfy 
\[ 
\sum_{i = 1}^\infty
\left( a_i^2 + b_i^2 - 1 \right)^2  < \infty. 
\]
\end{enumerate}
\end{theorem}
\begin{proof}
See \ref{sec:proof:t1} for the detailed proof.
\end{proof}
We note that Condition (2) of Theorem \ref{theo:1} is designed for situations where the parameter $u$ is partially updated, conditioned on those components corresponding to $b_i = 0$.
The other conditions in Theorem \ref{theo:1} ensure that the operators $\cA$, $\cB$, and $\cG$ are bounded, and that the measures $\nu$ and $\nu^\bot$ in \eqref{eq:nus} are equivalent.

To simplify subsequent derivations of acceptance probabilities and various examples of operator-weighted proposals, we now define the following useful transformations.
\begin{definition}[Transformation that diagonalizes the prior covariance operator]
\label{def:trans_v}
\begin{equation}
v = \prcov^{-\half} (u - u_{\rm ref}),
\label{eq:trans_v}
\end{equation} 
where $u_{\rm ref} = \prmean + m_{\rm ref}$ with $m_{\rm ref} \in \im(\prcov^{\half})$.
\end{definition}
\begin{definition}[Transformation that diagonalizes operators $\cA$, $\cB$, and $\cG$]
\label{def:trans_w}
\begin{equation}
w = \basis^{\ast} v.
\label{eq:trans_w}
\end{equation}
\end{definition}

By applying the transformation \eqref{eq:trans_v} to $u$, which has prior measure $\mu_0 = \normal(\prmean,\prcov)$, we obtain a prior measure on $v$ of the form $\tilde{\mu}^v_0 \equiv \normal(-v_{\rm ref},\cI)$, where $v_{\rm ref} = \prcov^{-\half}m_{\rm ref}$.
We can then define a reference measure $\mu^v_0=\normal(0,\cI)$ associated with the parameter $v$. 
Since we require $m_{\rm ref} \in \im(\prcov^{\half})$, by the Feldman-Hajek theorem (see \ref{sec:FH_theorem}) we have
$\tilde{\mu}^v_0 \ll \mu^v_0$ with 
\[
\frac{d\tilde{\mu}^v_0}{d\mu^v_0}(v) = 
\exp\left(- \frac12 \|v_{\rm ref}\|^2 - \langle v_{\rm ref} , v\rangle \right).
\]
The potential function $\potential$ has  input $u = \prcov^{\half} v + u_{\rm ref}$, which is by definition almost surely in $\hilbert$.
Therefore, $\mu_y^v \ll \mu^v_0$ by transitivity of absolute continuity. 
And thus the posterior measure defined on the parameter $v$ has the form: 
\begin{equation}
\frac{d \mu_y}{d \mu_0}(v) \propto \exp\left(-\potential(v ; y) - \frac12 \|v_{\rm ref}\|^2 - \langle v_{\rm ref}, v \rangle \right),
\label{eq:tranpost}
\end{equation}
where $\potential(v;y) = \potential(\prcov^{\half} v + u_{\rm ref} ; y)$. 

The proposal distribution \eqref{eq:operator_weighted} can then be rewritten as
\begin{equation}
v' =  \cA v - \cG \gradientv + \cB \rand ,
\label{eq:operator_weighted_v}
\end{equation}
where $\gradientv = \prcov^{\half} \gradientu$.
Applying the unitary transformation \eqref{eq:trans_w}, we can simultaneously diagonalize the operators $\cA$, $\cB$, and $\cG$. 
The resulting proposal distribution on $w$ is then
\begin{equation}
w' =  \cD_\cA w - \cD_\cG \gradientw + \cD_\cB \rand ,
\label{eq:operator_weighted_w}
\end{equation}
where we again use the degenerate notation $\gradientw = \basis^{\ast} \gradientv$, with $v = \basis w$.

\begin{remark}
Covariance operators in infinite dimensions are typically taken to be trace-class on $\hilbert$, so that draws have finite $\hilbert$-norm. For the transformed parameters $v$ and $w$, the associated proposals use a Gaussian measure with identity covariance operator.
Even though the identity operator is not trace-class on $\hilbert$ and hence draws from this Gaussian measure
do not have finite $\hilbert$ norm, they are still square-integrable in the weighted space $\im(\prcov^{-\half})$.
Furthermore, \eqref{eq:operator_weighted_v} or \eqref{eq:operator_weighted_w} still yield a well-defined function space proposal for the parameter $u$ after the inverse transformation is applied.
\end{remark}

For brevity of notation, the rest of this paper will define our proposals using the transformed parameter $v$. 
Dealing explictly with the transformed parameter $v$ also allows for a more computationally efficient implementation of the operator-weighted proposal in finite dimensions.
This is because each update only requires applying the operator $\prcov^{\half}$ once to transform the sample $v$ back to $u$ for evaluating the forward model, and another time to compute $\gradientv$. 
In comparison, the operator-weighted proposal defined for parameter $u$ requires applying the operator $\prcov^{\half}$ four times and applying the operator $\prcov^{-\half}$ once for each update.

The following corollary gives the acceptance probability of proposal \eqref{eq:operator_weighted}.
\begin{corollary}
\label{coro:1}
Consider the proposal
\[
v' = \cA v - \cG \gradientv + \cB \rand,
\]
where $\cA$, $\cB$, and $\cG$ satisfy the conditions of Theorem \ref{theo:1}.
The acceptance probability of this proposal is 
\[
\alpha(v, v') = \min \left\{ 1, \exp\left( \rho(v',v) - \rho(v,v') \right) \right\},
\]
where
\begin{eqnarray*}
\rho(v, v') & = & - \potential(v; y) - \left\langle v_{\rm ref}, v \right\rangle 
- \half \left\langle v, {\cB_0}^{-2}\left({\cA_0}^2 + {\cB_0}^2 - \cI_0 \right) v  \right\rangle \\
&& - \left \langle {\cB_0}^{-1} \cG_0  \gradientv , {\cB_0}^{-1} \left( \cI_0 v' - \cA_0 v \right) \right \rangle - \frac{1}{2} \left \| {\cB_0}^{-1} \cG_0 \gradientv \right \|^2  ,
\end{eqnarray*}
and where ${\cB}_0$, $\cA_0$, $\cI_0$, $\cG_0$ denote the respective operators projected onto $\{\psi_i\}_{b_i \neq 0}$.
\end{corollary}
\begin{proof}
The detailed proof is given in \ref{sec:proof:c1}.  
\end{proof}

Theorem \ref{theo:1} and Corollary \ref{coro:1} provide a wide range of possibilities for constructing valid 
operator-weighted proposals for infinite-dimensional problems.
As described in Section~\ref{sec:bkgd}, one important advantage of operator-weighted proposals over simple pCN proposals is the ability to prescribe different step sizes in different directions of the parameter space. 
By utilizing information from the Hessian of the data-misfit functional and the prior covariance, one can construct operators $\cA$, $\cB$, and $\cG$ that adapt to the structure of the posterior.
%

%% file: lis_r1.tex

\subsection{Likelihood-informed subspace}
\label{sec:scale_lis}

\added{
The first step in constructing the operators appearing in proposal distribution \eqref{eq:operator_weighted} is to identify a suitable set of basis functions $\basis$. We will choose these basis functions in order to expose a posterior structure wherein departures from the prior may be limited to a low-dimensional subspace of the parameter space. This structure is a consequence of several factors typical of inverse problems and other distributed-parameter inference configurations:  the smoothing action of the forward operator, the limited accuracy or number of observations, and smoothing from the prior.
The parameter subspace capturing this prior-to-posterior update will be called the {\it likelihood-informed subspace} (LIS). 
In designing function-space MCMC proposals, identification of the LIS will provide a useful way to adapt to posterior structure through the construction of the operators $\cA$, $\cB$, and $\cG$ in \eqref{eq:operator_weighted}.

For finite-dimensional inverse problems, the LIS can be derived from the dominant eigenvectors of the prior-preconditioned Hessian of the data misfit function, which is a finite dimensional version of the potential $\potential(u; y)$ in \eqref{eq:loglike}.
This construction was used in \cite{Flath_etal_2011} to build low-rank update approximations to the posterior covariance matrix for large-scale linear inverse problems. 
It has also been generalized to inverse problems with infinite-dimensional parameters, for the purpose of constructing a Laplace approximation of the posterior measure \cite{Bui_etal_2012, Bui_InfLinear_2013}. \cite{Linear_Redu_2014} proved the optimality of this approximation in the linear-Gaussian case---in the sense of minimizing the F\"{o}rstner-Moonen \cite{Forstner} distance from the \textit{exact} posterior covariance matrix over the class of positive definite matrices that are rank-$r$ negative semidefinite updates of the prior covariance, for any given $r$. These optimality results also extend to optimality statements between Gaussian distributions~\cite{Linear_Redu_2014}. 

The investigations mentioned above compute the LIS {\it locally}, through the eigendecomposition of the prior-preconditioned Hessian at a given parameter value $u$. 
When the forward model is nonlinear, however, the Hessian of the potential varies over the parameter space and hence a local LIS may not be sufficient to capture the entire prior-to-posterior update. 
One approach for dealing with a non-constant Hessian,  introduced in \cite{Cui_LIS_2014}, is to ``globalize'' the LIS by combining local LIS information from many points in the posterior.
\cite{Cui_LIS_2014}  uses this global LIS to approximate a non-Gaussian posterior in finite dimensions, where the posterior distribution in the complement of the global LIS is taken to be equal to the prior. In other words, the approximation decomposes the posterior into a finite subspace of likelihood-informed directions and a \textit{complementary subspace} (CS) of prior-dominated directions, with mutually independent distributions in each subspace.
Here we will extend the notion of a global LIS to the infinite-dimensional setting, where the CS becomes infinite-dimensional. 
Then we will incorporate the global LIS into the operator-weighted MCMC presented in Section~\ref{sec:op_weighted}, in order to achieve exact posterior sampling rather than posterior approximation.  

In the remainder of Section~\ref{sec:scale_lis}, we review the computation of the local LIS and then discuss the globalization of the LIS through the estimation of posterior expectations.
Based on the global LIS, we will also present a structure-exploiting posterior covariance approximation for non-Gaussian posteriors.
This covariance approximation naturally integrates the global LIS with the second-order geometry of the posterior, and will be used in later subsections to design operator-weighted proposals. 
We conclude this subsection by discussing iterative construction procedures for the global LIS and the posterior covariance approximation. 
}

\subsubsection{Local likelihood-informed subspace}
\label{sec:scale}

Let the forward model be Fr{\'e}chet differentiable. Then the linearization of the forward model at a given parameter $u$, $\linear(u) = \nabla_u \forward(u)$, yields the local sensitivity of the parameter-to-data map.
\added{The Gauss-Newton approximation of the Hessian of the data-misfit functional \eqref{eq:loglike} at $u$ (hereafter referred to as the GNH) then becomes}
\begin{equation}
\hessian(u) = \linear^\ast (u) \obscov^{-1} \,  \linear(u),
\label{eq:gn_hessian}
\end{equation}
and can be used to construct a local Gaussian approximation $\normal(m(u), \pocov(u))$ of the posterior measure, where  
\begin{equation}
\pocov^{-1}  (u)= \hessian(u) + \prcov^{-1} \quad {\rm and} \quad
m(u) = u - \pocov(u) \gradientu.  
\label{eq:pos_cov}
\end{equation}
Along a function $\varphi \in \im(\prcov^\half)$, consider the local Rayleigh ratio
\begin{equation}
\ray(\varphi;u) = \frac{\langle \varphi, \hessian(u) \varphi\rangle}{ \langle \varphi,  \prcov^{-1} \varphi \rangle} .
\label{eq:ray_quo_u}
\end{equation}
This ratio quantifies how strongly the likelihood constrains variation in the $\varphi$ direction relative to the prior. When this ratio is large, the likelihood limits variation in the $\varphi$ direction more strongly than the prior, and vice-versa. 
We note that similarly to the Onsager-Machlup functional, \eqref{eq:ray_quo_u} can be extended to $\hilbert$ by setting $R(\varphi;u)=0$ outside $\im(\prcov^{\half})$.  We will in particular be concerned with $\varphi$ such that $\|\prcov^{-\half}\varphi\|=1$.

Applying the change of variables $\phi = \prcov^{-\half}\varphi \in \hilbert$, the Rayleigh ratio \eqref{eq:ray_quo_u} becomes
\begin{equation}
\widetilde{\ray}(\varphi;u) = \ray(\phi ;u) = \frac{ \left \langle \prcov^{\half} \phi, \hessian(u) \prcov^{\half} \phi \right \rangle}{ \left \langle \phi, \phi \right \rangle }.
\label{eq:ray_quo_v}
\end{equation}
This naturally leads to considering the eigendecomposition of the prior-preconditioned GNH (ppGNH) 
\begin{equation}
\left(\prcov^{\half} \hessian(u) \prcov^{\half} \right) \phi_i = \lambda_i \phi_i,
\label{eq:eig}
\end{equation}
which is positive semidefinite by construction. 
The dominant eigenfunctions of the ppGNH are local maximizers of the Rayleigh quotient \eqref{eq:ray_quo_v}. Indeed, for the basis function $\varphi_i = \prcov^\half \phi_i \in \im(\prcov^\half)$, the Rayleigh quotient has value $\lambda_i$, i.e., 
\[
\ray(\varphi_i ;u) = \lambda_i.
\]
Thus, the ppGNH captures the (local) balance of likelihood and prior information described above: the largest eigenvalues of \eqref{eq:eig} correspond to directions along which the likelihood dominates the prior, and the smallest eigenvalues correspond to directions along which the posterior is determined almost entirely by the prior. 
The basis functions $\{\varphi_1, \ldots, \varphi_l\}$ corresponding to the $l$ leading eigenvalues of \eqref{eq:eig}, such that $\lambda_1 \geq \lambda_2 \geq \ldots \geq \lambda_l \geq \eigentrunc_{\rm l} > 0$, span the \textit{local} likelihood-informed subspace (LIS).
The resulting approximation of the local posterior covariance $\pocov(u)$ has the form
\begin{equation}
\pocov(u) \approx \pocov^{(l)} =  \prcov - \sum_{i = 1}^{l} \frac{\lambda_i}{\lambda_i+1}\varphi_i^{} \varphi_i^\ast \, ,
\label{eq:opt_linear}
\end{equation}
where $\varphi_i$ and $\lambda_i$ depend on the point $u$. Note that a Gaussian measure with this covariance and with mean $m(u)$ \eqref{eq:pos_cov} still endows the Hilbert space $\hilbert$ with full measure. In a finite-dimensional setting, this covariance corresponds precisely to the approximation shown to be optimal in \cite{Linear_Redu_2014}. Note also that the number of nonzero eigenvalues in \eqref{eq:eig} is most $N_y$, and thus the local LIS is indeed low-dimensional compared to the infinite-dimensional parameter space. If the nonzero eigenvalues $\lambda_i$ decay quickly, then one can reduce the dimension of the local LIS even further, with an appropriate choice of truncation threshold $\eigentrunc_{\rm l} > 0$.

\begin{remark}
\label{rem:whitening}
The change of variables $\phi = \prcov^{-\half} \varphi$ has the same role as the whitening transformation in Definition \ref{def:trans_v}.
For the parameter $v = \prcov^{\half}(u - u_{\rm ref})$, the GNH w.r.t.\ $v$ is
\[
\hessian(v) = \prcov^{\half} \hessian(u) \prcov^{\half},
\]
which is identical to the ppGNH in \eqref{eq:eig}.
The associated local Rayleigh quotient for parameter $v$ has the form
\begin{equation}
\widetilde{\ray}(\phi ;v) = \frac{\langle \phi, \hessian(v) \phi \rangle}{ \langle \phi, \phi \rangle }.
\label{eq:ray_quo_local}
\end{equation}
In this case, the orthonormal basis $\{ \phi_1, \ldots, \phi_l \}$ forms the local LIS for the parameter $v$.
Henceforth, for brevity and to maintain notation consistent with our exposition of operator-weighted proposals above, we will discuss the global LIS and the associated posterior covariance approximation in terms of the transformed parameter $v$.
\end{remark}

\subsubsection{Global likelihood-informed subspace}
\label{sec:lis}

To extend the pointwise criterion \eqref{eq:ray_quo_local} into a global criterion for identifying likelihood-informed directions for nonlinear problems, we consider the posterior expectation of the local Rayleigh quotient,
\[
\expect\left[\widetilde{\ray}(\phi; v)\right] = \frac{\langle \phi, S \phi\rangle}{ \langle \phi, \phi \rangle},
\]
where $S$ is the expectation of the local GNH (working in whitened coordinates):
\begin{equation}
\cS = \int_{\im(\prcov^{-\half})} \hessian(v) \mu_y(dv).
\label{eq:expected_hessian}
\end{equation}
\added{As in the local case, it is natural to then construct a global LIS from the eigendecomposition
\begin{equation}
\cS \,\theta_j = \rho_j \theta_j.
\label{eq:glis}
\end{equation}
The eigenfunctions corresponding to the $r$ leading eigenvalues $\rho_1 \geq \rho_2 \geq \ldots \geq \rho_{r} \geq \eigentrunc_{\rm g} > 0$, for some threshold $\eigentrunc_{\rm g} > 0$, form a basis $\Theta_r = [\theta_1, \ldots, \theta_r]$ of the global LIS (again in the whitened coordinates).
We typically choose a truncation threshold $\eigentrunc_{\rm g} \approx 0.1$, thus retaining parameter directions where the ``impact'' of likelihood is at least one tenth that of prior. 
The expected GNH can be approximated using a Monte Carlo estimator based on samples adaptively drawn from the posterior. This adaptive procedure is described in Section~\ref{sec:adaptsamp}.
}

\begin{remark}
Since we assume that the Fr{\'e}chet derivative of the forward operator exists, the GNH operator is defined on function space. We take the approach of differentiating and then discretizing, so that dimension independence is conserved assuming a convergent discretization of the linearization.  
\end{remark}

\added{
\subsubsection{Low-rank posterior covariance approximation}
\label{sec:cov}

The posterior covariance describes important aspects of the global structure of the posterior (e.g., its scale and orientation) and thus is often used in constructing efficient MCMC samplers for the finite dimensional setting---for example, the adaptive Metropolis \cite{Haario_2001} and adaptive Metropolis-adjusted Langevin \cite{Atchade_2006} methods.
When the posterior covariance matrix arises from the discretization of an infinite-dimensional posterior, however, computing, storing, and factorizing this matrix---all standard operations required by most MCMC methods---are not feasible.
Suppose that the discretized parameter space has dimension $N$.  Then estimating the posterior covariance from samples (without assuming any special structure) requires at least $O (N )$ posterior samples; the storage cost is $O \! \left (N^2 \right )$; and the computational cost of a covariance factorization is $O \! \left (N^3 \right )$. 
Here, we will instead use the global LIS to derive a covariance approximation scheme for which the costs of estimating, storing, and factorizing the posterior covariance all scale linearly with the discretization dimension $N$. 

Since the global LIS captures those parameter directions where the posterior differs most strongly from the prior, we can approximate the posterior measure by projecting the argument of the likelihood function onto the finite-dimensional global LIS spanned by $\Theta_r$.
Given the rank-$r$ orthogonal projection operator $\Pi_r^{} = \Theta_r^{} \Theta_r^\ast$, the approximation takes the form 
\begin{equation}
\mu_y(dv) \approx  \tilde{\mu}_y(dv) \propto \lkd(\Pi_r v; y) \mu_0(dv).
\label{eq:approx_post}
\end{equation} 
The orthogonal projector $\Pi_r$ decomposes the parameter as
\[
v = \Pi_r v + \left(\cI - \Pi_r \right)v,
\]
where $\Pi_r v$ and $\left(\cI - \Pi_r \right)v$ are projections of the parameter $v$ onto the low-dimensional LIS and the infinite-dimensional CS, respectively. 
Since the prior measure $\mu_0(dv)$ is Gaussian with zero mean and identity covariance, $\Pi_r v$ and $\left(\cI - \Pi_r \right)v$ are independent under the prior. 
Furthermore, the argument of the likelihood function in \eqref{eq:approx_post} depends only on the LIS-projected parameter $\Pi_r v$. 
Therefore, $\Pi_r v$ and $\left(\cI - \Pi_r \right)v$ are also independent under the approximated posterior $\tilde{\mu}_y(dv)$.
This independence and the LIS projection lead to an approximate posterior covariance (for the parameter $v$) that takes the form 
\[
{\rm Cov}_{\mu_y} [ v ] \approx \Sigma := {\rm Cov}_{\mu_y} [ \Pi_r v ] + {\rm Cov}_{\mu_0} [ \left(\cI - \Pi_r \right) v ].
\]
This covariance $\Sigma$ is the sum of the marginal posterior covariance within the LIS,
\[
{\rm Cov}_{\mu_y} [ \Pi_r v ] = \Theta_r^{} \Sigma_r^{} \Theta_r^\ast,
\]
where the $r\times r$ matrix $\Sigma_r$ is the posterior covariance projected onto the LIS basis $\Theta_r$, and the marginal prior covariance within the CS,
\[
{\rm Cov}_{\mu_0} [ \left(\cI - \Pi_r \right) v ] = \cI - \Theta_r^{} \Theta_r^\ast.
\]
Introducing the eigendecomposition $\Sigma_r^{} = \cW_r^{} \cD_r^{} \cW_r^\ast$ and reweighing the basis $\Theta_r$ using the orthogonal matrix $\cW_r$,
\begin{equation}
\basis_r = \Theta_r \cW_r, 
\label{eq:reweigh_basis}
\end{equation}
we can diagonalize the covariance approximation and write it as a low-rank update of the identity:
\begin{eqnarray}
\Sigma & = & \Theta_r^{} \Sigma_r^{} \Theta_r^\ast + ( \cI - \Theta_r^{} \Theta_r^\ast ) \nonumber \\
& = & \basis_r^{} \left( \cD_r^{} - \cI_r^{} \right) \basis_r^\ast + \cI \, . 
\label{eq:low_rank_cov}
\end{eqnarray}
where $\cI_r$ denotes the $r\times r$ identity matrix. 
The corresponding approximate posterior covariance for the parameter $u$ is $\prcov^{\half} \,\Sigma \, \prcov^{\half}$.
Note that $\Pi_r^{} =\Theta_r^{} \Theta_r^\ast=\basis_r ^{} \basis_r^\ast$.

Given a basis for the LIS, the number of samples required to estimate the low-rank approximation \eqref{eq:low_rank_cov} now depends on the rank of the LIS rather than on the discretized parameter dimension $N$.
Furthermore, the cost of storing this approximated covariance is dominated by the storage of the LIS basis, which is proportional to the discretization dimension $N$.
The  cost of computing the eigendecomposition of $\Sigma_r$ is of order $r^3$, which is asymptotically negligible compared with the discretization dimension.
Similarly, the factorization of the diagonalized posterior covariance \eqref{eq:low_rank_cov} only involves the $r\times r$ diagonal matrix $\cD_r$, and hence these operations are also asymptotically negligible.
The dominant computational cost here is reweighing the LIS basis as in \eqref{eq:reweigh_basis}, which is proportional to the discretization dimension $N$. 

}

\added{
\subsubsection{Iterative construction}
\label{sec:adaptsamp}

The low-rank approximation of the posterior covariance \eqref{eq:low_rank_cov} will be used to design operator-weighted proposals. Yet building this covariance approximation requires posterior samples, both for computing the global LIS and for constructing the projected posterior covariance $\Sigma_r$. These computations must therefore be carried out on-the-fly during posterior sampling. 
Here we discuss two building block algorithms---for estimating the global LIS and the projected posterior covariance from a given set of posterior samples---which are later used in the adaptive sampling algorithms of Section~\ref{sec:algo_opt}.

We first consider the global LIS. Suppose that we have a set of $m$ posterior samples $v_k \sim \mu_y$ for $k = 1 \ldots m$. 
The expected GNH \eqref{eq:expected_hessian} could be directly approximated using the Monte Carlo estimator  
\begin{equation}
\cS_m = \frac{1}{m} \sum_{k = 1}^{m} \hessian(v_k).
\label{eq:mc_ppgnh}
\end{equation}
If the discretized parameter has dimension $N$, the cost of storing $\cS_m$ and the computational cost of updating it are both proportional to $N^2$.
Furthermore, the cost of a direct eigendecomposition of $\cS_m$ for computing the global LIS basis is of order $N^3$. 
Thus, it is not feasible to directly apply the Monte Carlo estimate and eigendecomposition here. 
Instead, we again use the low-rank structure of the problem to design an iterative updating scheme for computing and updating the eigendecomposition of $\cS_m$ as more posterior samples are drawn.

As discussed earlier for the local LIS, each local GNH $\hessian(v_k)$ is well captured by a low-rank decomposition
\begin{equation}
\hessian(v_k) \approx \sum_{i = 1}^{l(v_k)} \lambda_i(v_k) \phi_i(v_k) \phi_i^\ast(v_k), 
\label{eq:local_hessian}
\end{equation}
where $\lambda_i(v_k)$ and $\phi_i(v_k)$ are eigenvalues and eigenvectors depending on the parameter $v_k$.
The eigenspectrum is truncated at some local threshold $\eigentrunc_{\rm l}$, i.e., $\lambda_1(v_k) \geq \lambda_2(v_k) \geq \ldots \geq \lambda_{l(v_k)}(v_k) \geq \eigentrunc_{\rm l} > 0$. 
Similar to the global LIS truncation threshold, we typically choose the local truncation threshold $\eigentrunc_{\rm l} = \eigentrunc_{\rm g} =  0.1$.
Now let $\Lambda(v_k)$ be the $l(v_k)\times l(v_k)$ diagonal matrix such that $\Lambda_{ij}(v_k) = \delta_{ij} \lambda_i(v_k)$, and let $\Phi(v_k) = \left[\phi_1(v_k), \ldots, \phi_{l(v_k)}(v_k) \right]$ be the rank-$l(v_k)$ local LIS basis. 
Using 
matrix notation, the decomposition of the local GNH can be rewritten as
\[
\hessian(v_k) \approx  \Phi(v_k) \,\Lambda(v_k) \Phi^\ast(v_k).
\]
By computing the action of the local GNH on functions---each action requires one forward model evaluation and one adjoint model evaluation---the eigendecomposition can be computed using either Krylov subspace methods \cite{Golub_2012} or randomized algorithms \cite{SVD_HMT_2011, SVD_Liberty_etal_2007}.
This way, the number of forward and adjoint evaluations used to obtain \eqref{eq:local_hessian} is proportional to the target rank $l(v_k)$. 

For the current sample set $\{v_k\}_{k = 1}^{m}$, suppose that the expected GNH also admits a 
rank-$M$ approximation for some $M(m)$, in the form of 
\[
\cS_m^{} \approx \Theta_m^{} \Xi_m^{} \Theta_m^\ast,
\]
where $\Xi_m \in \R^{M\times M}$, with $[\Xi_m]_{ij} = \delta_{ij} \rho_i$ and $\Theta_m = [\theta_1, \ldots, \theta_{M}]$ containing the eigenvalues and eigenfunctions of $\cS_m$, respectively. 
To maintain an accurate representation of the expected GNH and its eigendecomposition during the adaptive construction, we truncate this eigendecomposition conservatively rather than applying the $\eigentrunc_{\rm g} = 0.1$ threshold used to truncate the global LIS. 
In this paper, we set the truncation threshold to $\eigentrunc_{\rm s} = 10^{-4}$. 
Given a new sample $v_{m+1}$, the eigendecomposition can be updated using the incremental SVD method \cite{inc_SVD}; this procedure is described in Steps 6--9 of Algorithm \ref{algo:adaptS}.
In the discretization of the infinite-dimensional problem, the bases $\Theta_m$ and $\Phi(v_{m+1})$ are represented by ``tall and skinny'' matrices, and thus the computational cost of this iterative update is proportional to the discretization dimension $N$.
The major computational cost in this procedure is the eigendecomposition of the local GNH---which involves a number of linearized forward model and adjoint model solves---in Step 2 of Algorithm \ref{algo:adaptS}.
To initialize the global LIS, we choose the MAP estimate (defined in Section~\ref{sec:bayes}) as the initial sample at which to decompose the local GNH.
\renewcommand{\algorithmicrequire}{\textbf{Input:}}
\renewcommand{\algorithmicensure}{\textbf{Output:}}
\begin{algorithm}[h]
 \begin{algorithmic}[1]
 \Require{$\;\;$ $\Theta_m$, $\Xi_m$, $v_{m+1}$} 
 \Ensure{$\Theta_{m+1}$, $\Xi_{m+1}$, $\Theta_r$, $\Xi_r$}
 \Procedure{UpdateLIS}{$\Theta_m$, $\Xi_m$, $v_{m+1}$, $\Theta_{m+1}$, $\Xi_{m+1}$, $\Theta_r$, $\Xi_r$} 
 \State Compute the local LIS basis through $\hessian(v_{m+1}) \approx \Phi(v_{m+1}) \,\Lambda(v_{m+1}) \, \Phi^\ast(v_{m+1})$
 \If{ $m = 0$ } 
 \State Initialize $\cS_1$ by setting $\Theta_1 = \Phi(v_1)$ and $\Xi_1 = \Lambda(v_1)$
 \Else
 \State Compute the QR decomposition $[\Theta_m , \Phi(v_{m+1}) ] = \cQ \cR$ 
 \State Compute the new eigenvalues through the eigendecomposition 
 \[
 \frac{1}{m+1} \cR \left[ \begin{array}{cc} m\,\Xi_m & 0 \\ 0 & \Lambda(v_{m+1}) \end{array} \right]\cR^\ast = \cW \, \Xi_{m+1} \, \cW^\ast
 \]
 \State Compute the new basis $\Theta_{m+1} = \cQ \cW$
 \State Return new decomposition $\Theta_{m+1}$ and $\Xi_{m+1}$ such that all $\rho_i \geq \eigentrunc_{\rm s}$
 \EndIf
 \State Return LIS basis $\Theta_r = \Theta_{m+1}(:, \, 1\!:\!r)$ and LIS eigenvalues $\Xi_r = \Xi_{m+1}(1\!:\!r,\, 1\!:\!r)$, \newline 
such that $\rho_1, \ldots, \rho_r \geq \eigentrunc_{\rm g}$ \Comment{MATLAB notation is used here}
 \EndProcedure
 
 \end{algorithmic}
 \caption{Incremental update of the expected GNH and global LIS. }
 \label{algo:adaptS}
\end{algorithm}

Algorithm \ref{algo:adaptS} provides an efficient way of updating the expected GNH given a new posterior sample $v_{m+1}$. In Section~\ref{sec:algo_opt} we will describe how to embed Algorithm \ref{algo:adaptS} into an adaptive framework for posterior exploration, where new samples for each update are generated by operator-weighted MCMC. This approach thus \textit{automatically} explores the landscape of the likelihood; depending on this landscape, the number of samples needed to give a reasonable estimate of the expected GNH (and in particular, its dominant eigendirections) may vary significantly. (For instance, if the log-likelihood is quadratic, only one sample is needed!) In general, to limit the total number of expensive eigendecompositions of the local GNH, we would like to terminate this adaptive construction process as soon as the estimated LIS stabilizes.

To this end, let $\{\Theta_{r}, \Xi_{r}\}$ and  $\{\Theta_{r^\prime},\Xi_{r^\prime}\}$ denote the global likelihood informed subspaces and associated eigenvalues evaluated at successive iterations of the LIS construction process. And let
\[
\cS_r^{} = \Theta_r^{}\, \Xi_r^{}\, \Theta_r^{\ast}, \quad {\rm and} \quad \cS_{r^\prime}^{} = \Theta_{r^\prime}^{}\, \Xi_{r^\prime}^{}\, \Theta_{r^\prime}^{\ast},
\]
denote the (truncated) expected GNHs at the same two iterations.
We use the F\"{o}rstner distance \cite{Forstner} between $(\cI + \cS_r) $ and $( \cI + \cS_{r^\prime} )$,
\begin{equation}
\label{eq:lis_conv}
d_{\mathcal{F}} \left( \cI + \cS_r , \, \cI +  \cS_{r^\prime} \right) = 
\left( \sum_{i = 1}^{\infty} \ln^2 \left( \lambda_i \left( \cI + \cS_r, \, \cI + \cS_{r^\prime} \right) \right) \right)^{\frac12},
\end{equation}
to monitor the converge of the LIS. 
Here $\lambda_i( \cA, \cB )$ is the $i$th generalized eigenvalue of the pencil $(\cA,\cB)$.
The F\"{o}rstner distance was introduced in  \cite{Forstner} as a metric for positive definite covariance matrices, and is also related to the Kullback-Leibler divergence or Hellinger distance between Gaussian distributions, as described in \cite{Linear_Redu_2014}. 
In the present setting, the infinite-dimensional generalized eigenproblem in \eqref{eq:lis_conv} can be reformulated as a finite-dimensional generalized eigenproblem, so that the F\"{o}rstner metric becomes directly applicable.
Since the expected GNHs $\cS_r$ and $\cS_{r^\prime}$ are symmetric and finite-rank operators, the corresponding operators $(\cI + \cS_r)$ and $( \cI + \cS_{r^\prime} )$ differ only on the finite dimensional subspace $\sub$ that is the sum of the ranges of $\cS_r$ and $\cS_{r^\prime}$.
In other words, there are only a finite number of nonzero terms in the summation \eqref{eq:lis_conv}. Letting $\sub = {\rm span}(\Theta_r) + {\rm span}(\Theta_{r^\prime})$, $d_{\mathcal{F}}$ can be equivalently and efficiently computed by projecting the operators $(\cI + \cS_r )$ and $( \cI + \cS_{r^\prime} )$ onto $\sub$. Recall that in the discretized space both $\Theta_r$ and $\Theta_{r^\prime}$ are ``tall and skinny'' matrices, and therefore the computational cost of evaluating $d_{\mathcal{F}}$ scales linearly with the discretization dimension $N$. 

An alternative criterion for monitoring the convergence of the global LIS is the weighted subspace distance from \cite{Distance_Li_2009}, which is used in \cite{Cui_LIS_2014}. In practice, we find that both criteria produce consistent results. 

\begin{algorithm}[h]
 \begin{algorithmic}[1]
 \Require{$\;\;$ $\Theta_r$, $\Sigma_r$} 
 \Ensure{$\Psi_{r}$, $\cD_r$}
 \Procedure{UpdateCov}{$\Theta_r$, $\Sigma_r$, $\Psi_{r}$, $\cD_r$} 
 \State Compute eigendecomposition 
 $\cW_r^{} \cD_r^{} \cW_r^\ast = \Sigma_r^{}$
 \State Return the reweighted LIS basis $\Psi_r = \Theta_r \, \cW_r$ and the diagonalized covariance $\cD_r$
 \EndProcedure
 \end{algorithmic}
 \caption{Update of the low-rank posterior covariance approximation.}
 \label{algo:updateCov}
\end{algorithm}

Given a global LIS basis $\Theta_r$ and an empirical estimate of the posterior covariance projected onto the LIS, $\Sigma_r$, the low-rank approximation \eqref{eq:low_rank_cov} of the posterior covariance can be constructed using the procedure described in Algorithm \ref{algo:updateCov}.
When the LIS is fixed, estimation of $\Sigma_r$ from posterior samples only requires storage of order $r^2$ for the covariance matrix itself and an empirical mean. Given a new posterior sample $v$, $\Sigma_r$ can be updated by a rank-one modification using the $r$-dimensional vector $\Theta_r^\ast v$. 
Since the factorization and inversion of the posterior covariance \eqref{eq:low_rank_cov} only involve the $r\times r$ diagonal matrix $\cD_r$, their computational costs are negligible compared with the discretization dimension $N$.
The computational costs of projecting the sample $v$, reweighing the LIS basis as in \eqref{eq:reweigh_basis}, and computing the eigendecomposition of $\Sigma_r$ scale as $N r$, $N r^2$, and $r^3$, respectively. 
The sample projection $\Theta_r^\ast v$ has negligible cost compared with other operations such as PDE solves, and hence does not require special treatment. 

\algnewcommand{\IIf}[1]{\State\algorithmicif\ #1\ \algorithmicthen}
\algtext*{EndFor}

\begin{algorithm}[h]
 \begin{algorithmic}[1]
 \Require{$\;\;$ $\Theta_r$, $\Sigma_r$, $\Theta_{r^\prime}$} 
 \Ensure{$\Sigma_{r^\prime}$}
 \Procedure{ProjectCov}{$\Theta_r$, $\Sigma_r$, $\Theta_{r^\prime}$, $\Sigma_{r^\prime}$} 
 \State $ \cT = \Theta^\ast_{r^\prime} \Theta^{}_{r^{}} $
 \State $ \Sigma_{r^\prime} = \cT ( \Sigma_r^{} - \cI_r^{}) \cT^\ast + \cI_{r^\prime} $
 \EndProcedure
 \end{algorithmic}
 \caption{Posterior covariance re-projection for each LIS update.}
 \label{algo:projcov}
\end{algorithm}

Finally, we note that each time the global LIS is updated, the estimate $\Sigma_r$ of the posterior covariance in the LIS must itself be updated. Suppose that the current LIS basis is $\Theta_r$ and that the associated projected posterior covariance is $\Sigma_r$. Given a new LIS basis $\Theta_{r^\prime}$, the new projected posterior covariance $\Sigma_{r^\prime}$ can be obtained by projecting the current posterior covariance estimate $\Sigma =  \Theta_r^{} ( \Sigma_r^{} - \cI_r^{}) \Theta_r^\ast + \cI $ onto the new LIS, as shown in Algorithm \ref{algo:projcov}.
The computational cost of this procedure is dominated by the multiplication of the LIS bases in step 2, which is proportional to $N r r^\prime$.

}

%% file: li_opt_r1.tex

\subsection{DILI proposals}
\label{sec:lis_operators}

\added{
Now we use the global LIS and the low-rank posterior covariance approximation \eqref{eq:low_rank_cov} to develop operator-weighted MCMC proposals that are efficient and that satisfy the conditions of Theorem~\ref{theo:1}, thus ensuring discretization-invariant mixing in a function space setting.

The low-rank posterior covariance $\Sigma = \basis_r^{} \left( \cD_r^{} - \cI_r^{} \right) \basis_r^{\ast} + \cI$ is entirely specified by the global LIS basis $\basis_r$ and the $r$-dimensional diagonal matrix $\cD_r$.
Given these objects, the transformed parameter space $\hilbert_v \equiv \im(\prcov^{-\half})$ can be partitioned as $\hilbert_v = \hilbert_r \oplus \hilbert_\perp$, where the LIS is $\hilbert_r = \rm span (\basis_r)$ and has finite Hilbert dimension, i.e.,  ${\rm card}(\basis_r) = r < \infty$.  
We can associate the infinite-dimensional CS\footnote{In whitened coordinates, i.e., in terms of the transformed parameter $v$, the CS is the \textit{orthogonal} complement of the LIS.} $\hilbert_\perp \equiv \left ( \hilbert_v \right )^\perp$ with a set of eigenfunctions 
$\basis_\perp = [\psi_{r+1}, \psi_{r+2},\ldots]$ such that $\hilbert_\perp = \rm span (\basis_\perp)$. $\basis = [\basis_r, \basis_\perp]$ thus forms a complete orthonormal system in $\hilbert_v$.
Recall that the LIS captures parameter directions where the posterior differs most strongly from the prior. This partition of the parameter space, along with the approximate posterior covariance $\Sigma$, can then be used to assign eigenvalues to the basis $\basis$ in order to construct the operators $\cA$, $\cB$, and $\cG$ in the proposal \eqref{eq:operator_weighted}.

An appropriate discretization of the Langevin SDE \eqref{eq:langevin} naturally leads to an eigenvalue assignment that satisfies Theorem \ref{theo:1}.
Following the sequence of transformations in Definitions \ref{def:trans_v} and \ref{def:trans_w}, the Langevin SDE \eqref{eq:langevin} preconditioned by the low-rank approximation of the posterior covariance $\prcov^{\half} \Sigma \prcov^{\half}$ (for the parameter $u$) takes the form
\begin{equation}
\left(\begin{array}{l} d w_r^{} \\ d w_\perp^{} \end{array}\right) = 
- \left(\begin{array}{l} \cD_r^{} \, w_r^{} \, d\tau \\ w_\perp^{} \, d\tau \end{array}\right) 
- \left(\begin{array}{l} \gamma_r^{}\, \cD_r^{} \left( {\basis_r^\ast} \, \gradientw \right) d\tau \\ \gamma_\perp^{} \, {\basis_\perp^\ast} \, \gradientw \, d\tau \end{array}\right) 
+ \left(\begin{array}{l} \sqrt{2\cD_r^{}} \, \white_r^{}(\tau)\\ \sqrt{2} \, \white_\perp^{}(\tau) \end{array}\right),
\label{eq:langevin_time}
\end{equation}
where $w_r^{} = {\basis_r^\ast}v$ and $w_\perp^{} = {\basis_\perp^\ast} v$ are coefficients of the projection of $v$ onto the LIS and CS, respectively. 
The two tuning parameters $\gamma_r$ and $\gamma_\perp$ take values $\{0, 1\}$ to switch between a Langevin proposal and random walk proposal in the LIS and CS.
Now, discretizing this SDE is equivalent to assigning eigenvalues to the operators $\cA$, $\cB$, and $\cG$ of the operator-weighted proposal \eqref{eq:operator_weighted}.
We will employ an inhomogeneous time discretization strategy: a time discretization scheme tailored to the non-Gaussian structure---that is not necessarily dimension-independent---is used in the finite-dimensional LIS, while a standard Crank-Nicolson scheme---to ensure dimension-independence---is used in the infinite-dimensional CS.

In the rest of this subsection, we will provide specific examples of these operator-weighted proposals.
We will prescribe two different time steps to discretize \eqref{eq:langevin_time}, denoted by $\dt_r$ and $\dt_\perp$.
All the proposals we present depend only on the global LIS basis $\basis_r$, the $r \times r$ diagonalized covariance $\cD_r$, and the time steps $\dt_r$ and $\dt_\perp$.
Given the partition defined by the LIS, we will also present extensions of those proposals to Metropolis-within-Gibbs schemes. 

}

\subsubsection{Basic proposals}
\label{sec:basic_opt}

Here we consider all-at-once proposals that update all the parameter in a single iteration. 
Two examples will be provided, one which does not use gradient information on the LIS $(\gamma_r = 0)$
and one which does $(\gamma_r = 1)$.  
On the CS, both proposals use the pCN update with $\gamma_\perp = 0$.
Thus the proposal in the CS is given by
\begin{equation}
w_\perp' = a_\perp w_\perp + b_\perp \rand_\perp,
\label{eq:cs_propose}
\end{equation}
where 
\[
\rand_\perp \sim \normal(0, \cI_\perp), \quad a_\perp = \frac{2 - \dt_\perp }{2 + \dt_\perp }, 
\quad {\rm and} \quad b_\perp = \sqrt{1 - {a_\perp}^2}.
\]
Here we use $\cI_\perp$ to denote the identity operators associated with $\hilbert_\perp$, whereas the identity operator associated with $\hilbert_r$ is an $r\times r$ identity matrix $\cI_r$. 
For the corresponding eigenfunctions $\basis_\perp$, the operators $\cA$, $\cB$, and $\cG$ have uniform eigenvalues $a_\perp$, $b_\perp$, and $0$, respectively. 

Various discretization schemes can be used for the Langevin SDE in the LIS.
Let the eigenfunctions $\basis_r$ be associated with sequences of eigenvalues 
$(a_i)$, $(b_i)$, and $(g_i)$, $i = 1,\ldots, r$, 
giving rise to the $r \times r$ diagonal matrices 
$\cD_{\cA_r}, \cD_{\cB_r}$, and $\cD_{\cG_r}$.
The operators $\cA$, $\cB$, and $\cG$ can then be written as
\begin{eqnarray}
\label{eq:finite_splitting}
\cA & = & \basis_r \left(\cD_{\cA_r} - a_\perp \cI_r\right) {\basis_r^\ast} + a_\perp \cI, \nonumber \\
\cB & = & \basis_r \left(\cD_{\cB_r} - b_\perp \cI_r\right) {\basis_r^\ast} + b_\perp \cI, \\
\cG & = & \basis_r \cD_{\cG_r} {\basis_r^\ast}. \nonumber
\end{eqnarray}
Theorem \ref{theo:1} and Corollary \ref{coro:1} can be used to show when the operator-weighted proposals given by \eqref{eq:finite_splitting} yield well-defined MCMC algorithms, and to derive the associated acceptance probability.

Our  first operator-weighted proposal employs a Crank-Nicolson scheme in the LIS, with $\gamma_r = 0$.
The resulting discretized Langevin SDE in the LIS has the form
\[
w_r' - w_r = - \displaystyle  \frac12 \cD_r (w_r' + w_r) \dt_r + \sqrt{2\cD_r\dt_r} \rand_r ,
\]
where $\rand_r \sim \normal(0, \cI_r)$.
Together with the CS proposal \eqref{eq:cs_propose}, we have the following complete proposal:

\begin{proposal}[\textbf{LI-Prior:} Likelihood-informed proposal that leaves the shifted prior distribution invariant]
\label{prop:lis_prior}
\[
v' = \cA v + \cB \rand ,
\]
where $\cA$, $\cB$, and $\cG$ are given by  \eqref{eq:finite_splitting}, with 
\begin{eqnarray}
\nonumber
\cD_{\cA_r} & = & \left(2 \cI_r + \dt_r \cD_r \right)^{-1} \left(2 \cI_r - \dt_r \cD_r \right), \\
\label{eq:dadb}
\cD_{\cB_r} & = & \left(\cI_r - \cD_{\cA_r}^2 \right)^\half, \\
\cD_{\cG_r} &= & 0.
\nonumber
\end{eqnarray}
Since this proposal has $\cD_{\cA_r}^2 + \cD_{\cB_r}^2 = \cI_r^{}$, the acceptance probability can be simplified to 
\begin{equation}
\alpha(v, v') = \min \left\{ 1, \exp\left(\potential(v) - \potential(v') + \langle v_{\rm ref}, v-v' \rangle \right) \right\},
\label{eq:std_acc}
\end{equation}
where $v_{\rm ref} = \prcov^{-\half} m_{\rm ref}$.
\end{proposal}

We can instead consider a mixed discretization scheme, which uses an explicit discretization of the Langevin proposal ($\gamma_r = 1$) in the LIS. The discretized Langevin SDE in the LIS has the form
\[
w_r' - w_r = - \dt_r \cD_r \left( w_r +{\basis_r^\ast} \gradientw \right) + \sqrt{2\cD_r\dt_r} \rand_r, 
\]
where $\rand_r \sim \normal(0, \cI_r)$.
Combined with the CS proposal \eqref{eq:cs_propose}, this leads to the proposal:
\begin{proposal}[\textbf{LI-Langevin:} Likelihood-informed proposal that uses the Langevin proposal in the LIS]
\label{prop:lis_mala}
\[
v' = \cA v - \cG \gradientv + \cB \rand ,
\]
where $\cA$, $\cB$, and $\cG$ are given by  \eqref{eq:finite_splitting}, with
\begin{eqnarray}
\nonumber
\cD_{\cA_r} & = &  \cI_r - \dt_r \cD_r, \\
\label{eq:redexlang}
\cD_{\cB_r} & = & \sqrt{2 \dt_r \cD_r}, \\
\cD_{\cG_r} &= & \dt_r \cD_r.
\nonumber
\end{eqnarray}
The acceptance probability of this proposal has the form 
\begin{equation}
\label{eq:lis_prop}
\alpha(v, v') = \min \left\{ 1, \exp(\rho(v', v) - \rho(v, v')) \right\},
\end{equation}
where $\rho(v, v')$ is given by 
\begin{eqnarray*}
\rho(v, v') & = & - \potential(v; y) - \langle v_{\rm ref}, v \rangle - \frac12 \left\| {\basis_r^\ast} v \right\|^2 \\ 
&& - \frac12 \left \| {\cD_{\cB_r}}^{-1} \left( {\basis_r^\ast} v' - \cD_{\cA_r} {\basis_r^\ast} v + \cD_{\cG_r} {\basis_r^\ast} \gradientv \right) \right \|^2,
\end{eqnarray*}
where $v_{\rm ref} = \prcov^{-\half} m_{\rm ref}$.
\end{proposal}

\subsubsection{Metropolis-within-Gibbs scheme}
\label{sec:mwg_opt}

Given a basis $\basis_r$ for the LIS, the parameter $v$ can be split as
\[
v = v_r + v_\perp, \quad v_r =  {\basis_r} w_r, \quad v_\perp = \basis_\perp w_\perp,
\] 
and hence a Metropolis-within-Gibbs update can be quite naturally introduced for $v_r$ and $v_\perp$.
This update consists of a pair of proposals:
\begin{subequations} 
\begin{eqnarray}
{\rm Update\;} v_r {\rm \; conditioned \; on\;} v_\perp & : & v' = \cA_r v - \cG_r \gradientv+ \cB_r \rand,\label{eq:gibbs1} \\
{\rm Update\;} v_\perp {\rm \; conditioned \; on\;} v_r & : & v' = \cA_\perp v + \cB_\perp \rand, \label{eq:gibbs2}
\end{eqnarray}
\end{subequations}
where $\rand \sim \normal(0, \cI)$. 
By setting $a_\perp = 1$ and $b_\perp = 0$ in  \eqref{eq:finite_splitting}, we can derive the proposal \eqref{eq:gibbs1}, which updates $v_r$ conditional on $v_\perp$.
This leads to the operators
\begin{equation}
\cA_r = \basis_r \left( \cD_{\cA_r} - \cI_r \right) {\basis_r^\ast} + \cI, \quad \cB_r = \basis_r \cD_{\cB_r} {\basis_r^\ast} , \quad {\rm and} \quad \cG_r = \basis_r \cD_{\cG_r} {\basis_r^\ast} .
\label{eq:mg1}
\end{equation}
Similarly, by setting $\cD_{\cA_r} = \cI_r$, $\cD_{\cB_r} = 0$, and $\cD_{\cG_r} = 0$, we have the proposal \eqref{eq:gibbs2} which updates $v_\perp$ conditional on $v_r$.
This leads to the operators 
\begin{equation}
\cA_\perp = \basis_r {\basis_r^\ast} + a_\perp \left( \cI - \basis_r {\basis_r^\ast} \right), \quad {\rm and} \quad \cB_\perp = b_\perp \left( \cI - \basis_r {\basis_r^\ast} \right).
\label{eq:mg2}
\end{equation}
These proposals both satisfy the conditions of Theorem \ref{theo:1}, 
and therefore lead to a Metropolis-Hastings scheme that is well-defined on function space. 
The following two Metropolis-within-Gibbs schemes are derived from Proposals \ref{prop:lis_prior} and \ref{prop:lis_mala}.

\begin{proposal}
[\textbf{MGLI-Prior:} Metropolis-within-Gibbs update using a likelihood-informed proposal that is invariant w.r.t.\ the shifted prior distribution]
\label{prop:mglis_prior}  
Let $\cD_{\cA_r}$, $\cD_{\cB_r}$, and $\cD_{\cG_r}$ be given by  \eqref{eq:dadb}.
Define $\cA_r$, $\cB_r$, and $\cG_r$ as in  \eqref{eq:mg1} and $\cA_\perp$,  $\cB_\perp$ as in  \eqref{eq:mg2}.  
The updates are given by \eqref{eq:gibbs1} and \eqref{eq:gibbs2} 
and both have acceptance probability \eqref{eq:std_acc}.
\end{proposal}

\begin{proposal}[\textbf{MGLI-Langevin:} Metropolis-within-Gibbs update using the Langevin proposal in the LIS] 
\label{prop:mglis_mala} 
Let $\cD_{\cA_r}$, $\cD_{\cB_r}$, and $\cD_{\cG_r}$ be given by  \eqref{eq:dadb}.
Define
$\cA_r$, $\cB_r$, and $\cG_r$ as in  \eqref{eq:mg1} and
$\cA_\perp$ and $\cB_\perp$ 
as in \eqref{eq:mg2}.  
The updates are given by  \eqref{eq:gibbs1} and \eqref{eq:gibbs2}, 
with acceptance probabilities given by  \eqref{eq:lis_prop} and 
\eqref{eq:std_acc}, respectively.
\end{proposal}

%% file: algo_opt_r1.tex
\added{
\subsection{Adaptive posterior sampling framework}
\label{sec:algo_opt}

Now we describe how the construction of the global LIS, the low-rank posterior covariance approximation, and the operator-weighted proposals can be integrated into a single adaptive posterior sampling framework. The building blocks have already been given; in particular, Section~\ref{sec:adaptsamp} describes how to construct and iteratively update the global LIS basis $\Theta_r$ (or equivalently $\basis_r$ \eqref{eq:reweigh_basis}) and the posterior covariance $\Sigma$ from posterior samples, while Section~\ref{sec:lis_operators} shows how to construct several different operator-weighted proposals from the LIS and covariance information. All that remains is to embed the first two operations within an MCMC approach. The overall framework is described in Algorithm~\ref{algo:adaptM}. At specified intervals (controlled by the parameter $n_{\text{lag}}$) and until a final tolerance (controlled by $\Delta_{\text{LIS}}$ and $n_{\text{max}}$) is reached, the LIS is updated from previous posterior samples. Every $n_{\text{b}}$ samples and/or when the LIS basis is updated, and for the duration of the algorithm, the posterior covariance $\Sigma$ and the operators in the DILI proposal are updated from previous samples. Thus the algorithm adaptively captures the global structure of the posterior from local Hessian information and from samples, and tunes the resulting operator-weighted proposals accordingly.

\algnewcommand{\IElse}{\algorithmicelse \unskip\ }
\algnewcommand{\EndIIf}{\unskip\ \algorithmicend\ \algorithmicif}
\begin{algorithm}[h]
 \begin{algorithmic}[1]
 \Require{During the LIS construction, we retain (1) $\{\Theta_m, \Xi_m\}$ to store the expected GNH evaluated from $m$ samples; and (2) the value of $d_\mathcal{F}$ \eqref{eq:lis_conv} between the most recent two updates of the expected GNH, for LIS convergence monitoring.}
 \Require{At step $n$, given the state $v_n$, LIS basis $\Theta_r$, projected empirical posterior covariance $\Sigma_r$, and operators $\{\cA, \cB, \cG\}$ induced by $\{\basis_r, \cD_r, \dt_r, \dt_\perp\}$, one step of the algorithm is:} 
 \State Compute a candidate $v' = q(v_n, \cdot\,;\, \cA, \cB, \cG)$ using either LI-prior or LI-Langevin
 \State Compute the acceptance probability $\alpha(v_n, v')$ 
 \IIf{${\rm Uniform}(0, 1] < \alpha(v_n, v')$} $v_{n+1} = v'$ 
 \IElse $v_{n+1} = v_n$
 \Comment{Accept/reject $v'$}
 \State $\textsc{UpdateFlag} = {\bf false}$
 \If{ ${\rm rem}(n\! + \!1, n_{\rm lag}) \! = \! 0 \And m \! < \! n_{\text{max}} \And d_\mathcal{F} \! \geq \!\Delta_{\rm LIS}$ }
 \Comment{Update LIS using selected samples}
 \State $\Call{UpdateLIS}{\Theta_m, \Xi_m, v_{n+1}, \Theta_{m+1}, \Xi_{m+1}, \Theta_{r'}, \Xi_{r'}}$
 \State $\Call{ProjectCov}{\Theta_r, \Sigma_r, \Theta_{r'}, \Sigma_{r'}}$
 \State Update the LIS convergence diagnostic $d_\mathcal{F}$ as in \eqref{eq:lis_conv}
 \State $\Theta_r = \Theta_{r'}$, $\Xi_{r} = \Xi_{r'}$, $\Sigma_r = \Sigma_{r'}$, $m = m+1$
 \State $\textsc{UpdateFlag} = {\bf true}$
 \Else \Comment{The LIS is fixed}
 \State Perform rank-$1$ update of the empirical covariance $\Sigma_r$ given $\Theta_r^* v_{n+1}^{}$
 \IIf{$ {\rm rem}(n\!+\!1, n_{\rm b}) = 0 $} $\textsc{UpdateFlag} = {\bf true}$ 
 \EndIf
 \If{$\textsc{UpdateFlag}$}
 \State $\Call{UpdateCov}{\Theta_r,\Sigma_r,\Psi_{r},\cD_r}$
 \State Update the operators $\{\cA, \cB, \cG\}$
 \EndIf
 \end{algorithmic}
 \caption{Adaptive function space MCMC with operator-weighted proposals.}
 \label{algo:adaptM}
\end{algorithm}

Generally speaking, the major computational cost of Algorithm \ref{algo:adaptM} is the evaluation of the acceptance probability at Step 2, which involves a computationally expensive likelihood evaluation and---when the Langevin proposals are used---one extra adjoint solve for the gradient evaluation. 
The proposal evaluation in Step 1 only incurs a minor computational cost, as the cost of applying the low-rank decompositions of the operators $\cA$, $\cB$, and $\cG$ is proportional to the discretization dimension $N$.
The computational cost of the LI-Prior proposal is therefore comparable to that of a pCN random walk proposal---or other types of random walk proposals used in the finite dimensional setting---and the computational cost of the LI-Langevin proposal is comparable to standard Metropolis-adjusted Langevin proposals. 
Purely local methods, such as stochastic Newton \cite{Martin_2012} and manifold MCMC \cite{Girolami_2011}, are much more computationally expensive than both LI-Prior and LI-Langevin, as they require high-order derivatives of the posterior at every step. 

Another non-negligible computational burden comes from the adaptive construction of LIS (Steps 5--10), where each iteration involves a number of linearized forward model and adjoint model evaluations for computing the eigendecomposition of the local GNH, as well as the incremental update of the factorized expected GNH in Algorithm \ref{algo:adaptS}. 
Here we employ a fixed-size adaptation strategy that limits the total number of local GNH eigendecompositions to be no more than $n_{\text{max}}$, and we terminate the iterative updates once the subspace convergence diagnostic \eqref{eq:lis_conv} drops below a specified threshold $\Delta_{\rm LIS}$.
Furthermore, to ensure that essentially uncorrelated posterior samples are used to estimate the expected GNH, sub-sampling with lag $n_{\rm lag}$ is used in Step 7. 
Since the LIS basis is low-rank, the computational costs of updating the posterior covariance (Step 16) and the operators (Step 17) are only proportional to the discretization dimension $N$, and thus do not dominate the cost of the forward solve.
These costs are further suppressed by only renewing these two terms every $n_{\text{b}}$ iterations during MCMC sampling. In the numerical experiments below, we choose $n_{\text{b}} = 50$.

We also note that once the additive construction of LIS terminates, we simulate the Markov chain using the operator-weighted proposal built from a \textit{fixed} LIS, where the empirical covariance is henceforth updated only within this fixed $r$-dimensional subspace. 
In this situation, existing adaptive MCMC theory \cite{Haario_2001,  AM_adapt_2006, Atchade_2006, Roberts_2007, saksman2010ergodicity} can be used to establish ergodicity of the chain. 

\begin{algorithm}[h]
 \begin{algorithmic}[1]
 \Require{During the LIS construction, we have $\{\Theta_m, \Xi_m\}$ and $d_\mathcal{F}$ as in Algorithm \ref{algo:adaptM}.}
 \Require{At step $n$, given the state $v_n$, LIS basis $\Theta_r$, projected empirical posterior covariance $\Sigma_r$, and operators $\{\cA_r, \cB_r, \cG_r\}$ and $\{\cA_\perp, \cB_\perp\}$ induced by $\{\basis_r, \cD_r, \dt_r, \dt_\perp\}$, one step of the algorithm is:} 
 \State Compute a LIS candidate $v' = q(v_n, \cdot\,;\, \cA_r, \cB_r, \cG_r)$ using either MGLI-Prior or MGLI-Langevin
 \State Compute the acceptance probability $\alpha(v_n, v')$ 
 \IIf{${\rm Uniform}(0, 1] < \alpha(v_n, v')$} $\tilde{v} = v'$ 
 \IElse $\tilde{v} = v_n$
 \EndIIf
 \Comment{Accept/reject $v'$}
  \State Compute a CS candidate $v'' = q(\tilde{v}, \cdot\,;\, \cA_\perp, \cB_\perp)$ using \eqref{eq:gibbs2}
 \State Compute the acceptance probability $\alpha(\tilde{v}, v'')$ 
 \IIf{${\rm Uniform}(0, 1] < \alpha(\tilde{v}, v'')$} $v_{n+1} = v''$ 
 \IElse $v_{n+1} = \tilde{v}$
 \Comment{Accept/reject $v''$}
 \State Follow Steps 4--14 of Algorithm \ref{algo:adaptM} to update the LIS and sample covariance within the LIS  
 \If{$\textsc{UpdateFlag}$}
 \State $\Call{UpdateCov}{\Theta_r,\Sigma_r,\Psi_{r},\cD_r}$
 \State Update the operators $\{\cA_r, \cB_r, \cG_r\}$ and $\{\cA_\perp, \cB_\perp\}$
 \EndIf
 \end{algorithmic}
 \caption{Adaptive function space MCMC with operator-weighted proposals, using Metropolis-within-Gibbs updates.}
 \label{algo:adaptMWG}
\end{algorithm}

Just as the all-at-once proposals of Section~\ref{sec:basic_opt} are integrated into the adaptive framework of Algorithm~\ref{algo:adaptM}, the Metropolis-within-Gibbs proposals given in Section~\ref{sec:mwg_opt} can be integrated into Algorithm \ref{algo:adaptMWG}. Here the operators $\{\cA_r, \cB_r, \cG_r\}$ and $\{\cA_\perp, \cB_\perp\}$---for the LIS update and the CS update, respectively---are determined by $\{\basis_r, \cD_r, \dt_r, \dt_\perp\}$.
Two proposal and accept/reject steps are involved in each iteration of Algorithm \ref{algo:adaptMWG}, and thus it doubles the number of likelihood evaluations per iteration relative to Algorithm \ref{algo:adaptM}.
The computational cost of all the other components remains the same. 

}

\subsection{Benchmark algorithms}
\label{sec:bench}
In the numerical examples of Sections~\ref{sec:elliptic} and Section~\ref{sec:cd}, we will benchmark our operator-weighted proposals by comparing their performance with two representative proposals from the literature, one that is dimension independent and another that exploits posterior Hessian information.

\begin{proposal}[\textbf{pCN-RW:} Prior-preconditioned Crank-Nicolson random walk proposal \cite{BRSV_2008}]
\label{prop:pcn}
By setting $\gamma = 0$ in \eqref{eq:pcn} and applying the transformation in Definition \ref{def:trans_v}, we obtain the proposal: 
\[
v' = a v + \sqrt{1 - a^2} \rand .
\]
where $a \in (-1,1).$
\end{proposal}
Note that the gradient term could have been included by setting $\gamma = 1$ in \eqref{eq:pcn}. In the numerical examples described below, however, the performance of such a pCN Langevin proposal \cite{CRSW_2012} is very close to that of the pCN-RW proposal. For brevity, we will therefore only report on pCN-RW in our numerical comparisons.
\begin{proposal}[\textbf{H-Langevin:} Hessian-preconditioned explicit Langevin proposal] 
\label{prop:h_mala}
Let $\hessian_{\rm M}$ denote the Hessian of the data-misfit functional evaluated at the MAP estimate of the parameter $v$.
Precondition the Langevin SDE \eqref{eq:langevin} with the inverse Hessian of the OMF, $ \left ( \hessian_{\rm M} + \cI \right )^{-1}$, and discretize this SDE explicitly, arriving at a proposal of the form:
\begin{eqnarray*}
v' & = & \left( \cI - \dt (\hessian_{\rm M} + \cI)^{-1} \right) v - \dt (\hessian_{\rm M} + \cI)^{-1} D_v \Phi(v; y) + \sqrt{2 \dt (\hessian_{\rm M} + \cI)^{-1}} \rand \\
 & = & \cA_H v - \cG_H D_v \Phi(v; y) + \cB_H \rand.
\end{eqnarray*}
Now take a low rank approximation of $\hessian_{\rm M}$,
\[
\hessian_{\rm M} \approx \basis_{\rm M} \cD_{\rm M} {\basis_{\rm M}^\ast}.
\]
to obtain
\begin{eqnarray*}
\cA_H & = & \cI - \dt \left( \basis_{\rm M} \left(- \cD_{\rm M} (\cD_{\rm M} + \cI_r)^{-1} \right) {\basis_{\rm M}^\ast} + \cI \right), \\
\cB_H & = & \sqrt{2\dt} \left( \basis_{\rm M} \left( (\cD_{\rm M} + \cI_r)^{-\half} - \cI_r \right) {\basis_{\rm M}^\ast} + \cI \right), \\
\cG_H & = & \dt \left( \basis_{\rm M} \left( - \cD_{\rm M} (\cD_{\rm M} + \cI_r)^{-1} \right) {\basis_{\rm M}^\ast} + \cI \right) .
\end{eqnarray*}
\end{proposal}

This H-Langevin proposal essentially combines preconditioned MALA \cite{Roberts_2002} with the low-rank Hessian approximation used in the stochastic Newton algorithm \cite{Martin_2012}.
But instead of using a location-dependent preconditioner as in simplified manifold MALA~\cite{Girolami_2011} and stochastic Newton, it uses a constant preconditioner $(\hessian_{\rm M} + \cI)^{-1}$, thus avoiding the computational burden of evaluating the local Hessian at each iteration. 
The same principle has recently been used in \cite{Petra_stochnewton2013} to modify the stochastic Newton algorithm.
The main difference between the H-Langevin proposal and the variants of stochastic Newton proposed in \cite{Petra_stochnewton2013} is that H-Langevin employs an adjustable discretization step $\dt$, whereas \cite{Petra_stochnewton2013} uses $\dt = 1$.

%% file: elliptic_r1.tex

\section{Example 1: Elliptic PDE}
\label{sec:elliptic}

Our first numerical example is an inverse problem of an elliptic PDE, i.e., inferring the transmissivity field of the Poisson equation. 
We use this test case to compare the efficiency of the new operator-weighted proposals described in Section~\ref{sec:lis_operators} with the benchmark proposals of Section~\ref{sec:bench}.
We also demonstrate the advantages of using multiple posterior samples to construct the LIS and the invariance of the LIS under grid refinement.

\subsection{Problem setup}
Consider the spatial domain $\Omega = [0, 1]^2$  with boundary $\partial \Omega$. We denote the spatial coordinate by $s \in \Omega$. Let $\kappa(s)$ be the transmissivity field, $p(s)$ be the potential function, and $f(s)$ the forcing term. The potential function for a given realization of the transmissivity field is governed by 
\begin{equation}
\left\{ 
\begin{array}{lcll}
-\nabla \cdot \left( \kappa(s) \nabla p(s) \right) & = & f(s), & s \in \Omega \\
\left \langle \kappa(s) \nabla p(s), \vec{n}(s)  \right \rangle & = & 0, & s \in \partial \Omega
\end{array}
\right.
\label{eq:forward_e}
\end{equation}
where $\vec{n}(s)$ is the outward normal to the boundary. To make a well-posed boundary value problem, a further condition
\begin{equation}
\int_{\partial \Omega} p(s) d l(s)= 0,
\label{eq:bnd}
\end{equation}
is imposed.
The source/sink term $f(s)$ is defined by the superposition of four weighted Gaussian plumes with standard deviation $0.05$, centered at $[0.3, 0.3]$, $[0.7, 0.3]$, $[0.7, 0.7]$, $[0.3, 0.7]$, with weights $\{2, -3, -2, 3\}$.
The system of equations (\ref{eq:forward_e}) is solved by the finite element method with bilinear elements on a uniform $40\times 40$ grid. 

The transmissivity field is endowed with a log-normal prior distribution, i.e.,
\begin{equation}
\label{eq:prior_e}
\kappa(s) = \exp(u(s)), \; {\rm and} \; u(s) \sim \normal\left(0, \cC \right),
\end{equation}
where the covariance operator $\cC$ is defined through an exponential kernel function
\begin{equation}
\label{eq:corr}
c(s, s') = \sigma_u^2 \exp\left( - \frac{\|s - s'\|}{2 s_0} \right), \; {\rm for} \; s, s' \in \Omega.
\end{equation}
In this example, we set the prior standard deviation $\sigma_u = 1.25$ and the correlation length $s_0 = 0.0625$.
To make the inverse problem more challenging, we use a ``true'' transmissivity field that is not directly drawn from the prior distribution. The true transmissivity field---i.e., the field used to generate noisy observations---and the corresponding potential are shown in Figures \ref{fig:setup_e}(a) and (b), respectively. 

\begin{figure}[h!]
\centerline{\includegraphics[width=0.75\textwidth]{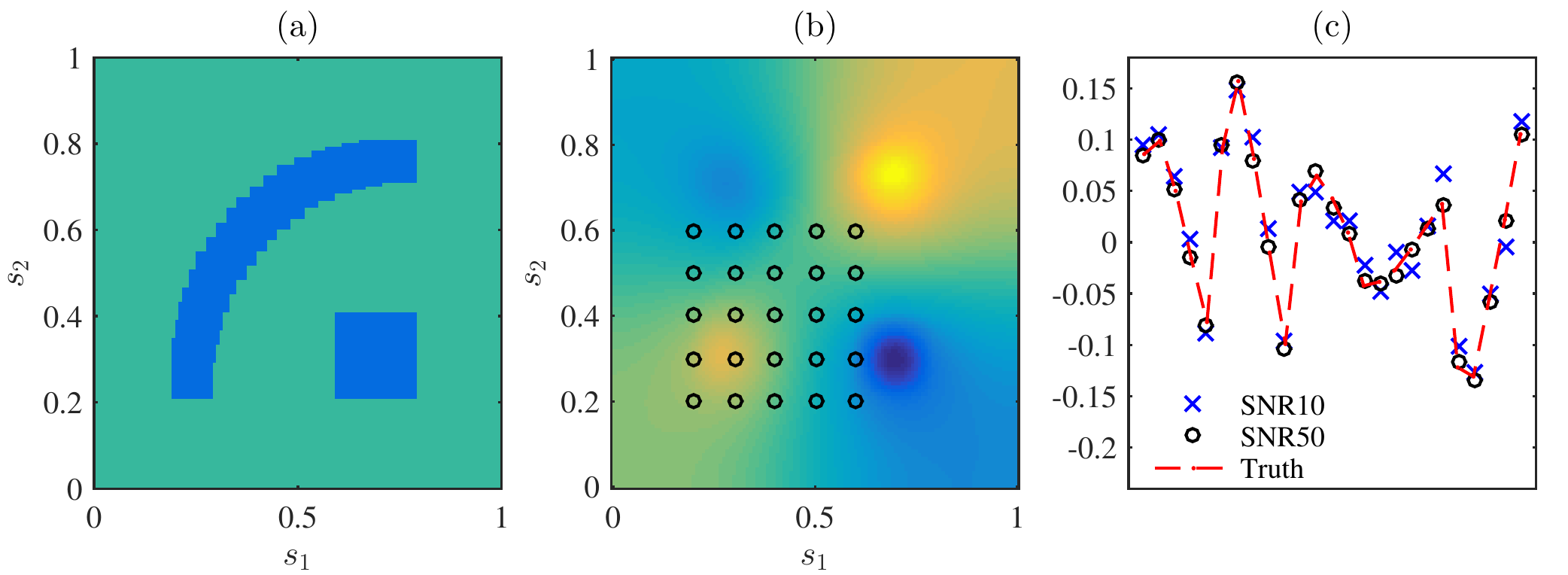}}
\caption{Setup of the elliptic example. (a) The true transmissivity field. (b) The potential $p(s)$ computed from the transmissivity field defined in (a), with measurement locations indicated by circles. (c) Data sets: the dashed line represents the noise-free potential at each measurement location, while crosses and circles are observations of the potential perturbed with Gaussian noise of signal-to-noise ratios 10 and 50, respectively.}
\label{fig:setup_e}
\end{figure}

Partial observations of the potential function are collected at $25$ measurement sensors as shown by the  circles in Figure \ref{fig:setup_e}(b).  The corresponding observation operator $\cM$ yields the outputs
\[
y = \cM p(s) + e, \quad e \sim \normal(0, \sigma^2 \cI_{25}),
\]
where the observed data is $y \in  \R^{25}$.  
Since amount of information carried in the likelihood function affects the structure of the LIS, we use two cases with different noise magnitudes to illustrate this relationship.
Let the signal-to-noise ratio (SNR) be defined as $\max_s\{u(s)\}/\sigma$.
The first data set is generated with signal to noise ratio 10 (SNR10), and the second set is generated with signal to noise ratio 50 (SNR50).
Both data sets are shown in Figure \ref{fig:setup_e}(c).
\added{We note that changing the signal-to-noise ratio will have an impact on the LIS dimension similar to changing the number of measurements; both affect the amount of information carried in the data.}

\subsection{Sampling efficiency}
\label{sec:samplingcompare}

We now compare the posterior sampling performance of our operator-weighted proposals (Proposals \ref{prop:lis_prior}--\ref{prop:mglis_mala}) with that of the pCN-RW and H-Langevin proposals (Proposals \ref{prop:pcn} and \ref{prop:h_mala}). To build the LIS, the adaptive sampling procedure in Section~\ref{sec:adaptsamp} is run for $2\times 10^5$ MCMC iterations, from which $1000$ posterior samples are selected, one every $n_{\rm lag} = 200$ iterations, to compute the expected Hessian. Then we sample the posterior distribution for $10^6$ MCMC iterations using the operator-weighted proposals and the benchmark proposals. 
Samples from the second half of these iterations are used to estimate the autocorrelation of the Onsager-Machlup functional \eqref{eq:omf} and of the parameter $u$ projected onto selected eigenfunctions of the prior covariance operator. As the metrics of sampling efficiency below will indicate, these sample sizes are significantly larger than one might use in practice, particularly with the DILI proposals. Our goal here, however, is to characterize the mixing properties of each chain without any potential bias resulting from burn-in and finite sample size. 
Also, our comparisons of mixing along particular directions in the parameter space use projections onto prior covariance eigenfunctions because these projections provide a \textit{common} basis for comparison---in contrast to projections $w$ onto modes of the LIS, which can be different for different proposals. Moreover, evaluating performance on LIS modes could conceivably yield results that are artificially favorable to the DILI proposals.

For the SNR10 case, Figure~\ref{fig:SNR10_v} shows the autocorrelation functions of selected components of the $v$ parameter. Note that we chose the square root  $\prcov^{\half}$ in Definition~\ref{def:trans_v} such that the selected components of $v$ correspond to projections of $u$ onto the 1st, 10th, 100th, and 1000th eigenfunctions of the prior covariance (i.e., Karhunen-Lo\`{e}ve modes). 
Overall, MGLI-Langevin outperforms the other proposals, in some cases quite dramatically.
On the low-index eigenfunctions (1 and 10), the operator-weighted proposals and the H-Langevin proposal yield similar mixing, as indicated by the decay of the autocorrelation. 
Yet the operator-weighted proposals outperform H-Langevin on the higher-index eigenfunctions, as these directions in the parameter space tend to be less constrained by the data; indeed, they are more aligned with the CS. This effect can already be observed for 100th eigenfunction and becomes more significant for the 1000th.
Chains produced by pCN-RW decorrelate very slowly compared to the other samplers.

The autocorrelation of the OMF, shown in Figure \ref{fig:SNR10_omf} along with trace plots for two of the chains, provides an alternative summary of sampling performance. Again, we observe that the operator-weighed proposals yield shorter autocorrelation times than H-Langevin and pCN-RW. 

Figure \ref{fig:SNR50_v} compares the MGLI-Langevin, H-Langevin, and pCN-RW samplers for the SNR50 case.
In parameter directions corresponding to the smoother prior covariance eigenfunctions (e.g., components 1, 10, and 100 of $v$), the MGLI-Langevin and H-Langevin proposals produce very similar autocorrelations.
Yet, as in the SNR10 case, MGLI-Langevin significantly outperforms H-Langevin for parameter directions corresponding to the high-index eigenfunctions, e.g., $v_{1000}$.
The pCN-RW chain again decorrelates very slowly compared to the other proposals. 
Looking at the OMF in Figure \ref{fig:SNR50_omf}, we observe relative performances similar to the SNR10 case: the MGLI-Langevin proposal outperforms both pCN-RW and H-Langevin.  

\begin{figure}[h!]
\centerline{\includegraphics[width=0.75\textwidth]{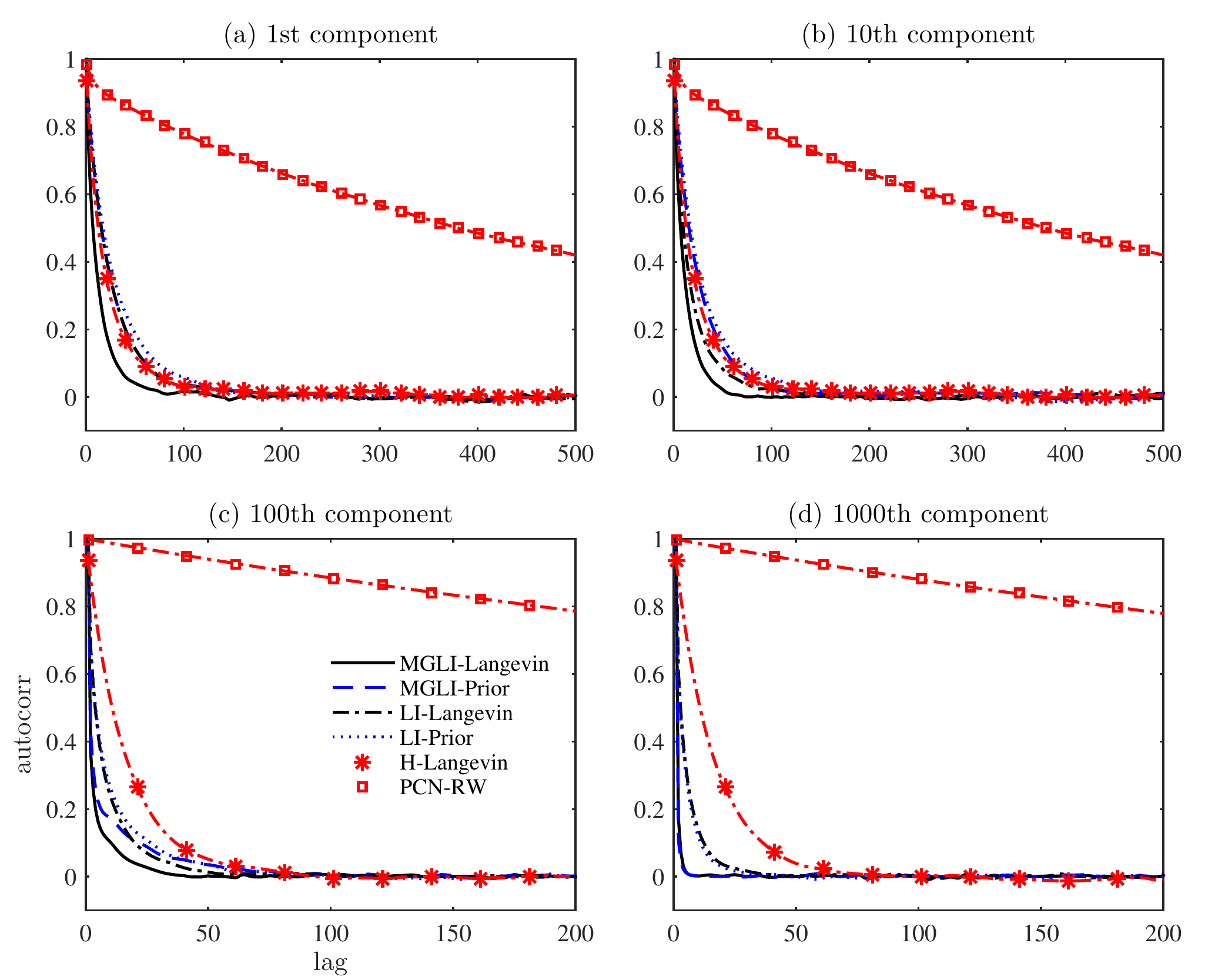}}
\caption{Autocorrelation functions of selected components of the $v$ parameter produced by various proposals, SNR10 case. (a) Projection onto the first eigenfunction of the prior covariance; (b) Projection onto the second eigenfunction; (c) Projection onto the 100th eigenfunction; (d) Projection onto the 1000th eigenfunction.}
\label{fig:SNR10_v}
\end{figure}

\begin{figure}[h!]
\centerline{\includegraphics[width=0.75\textwidth]{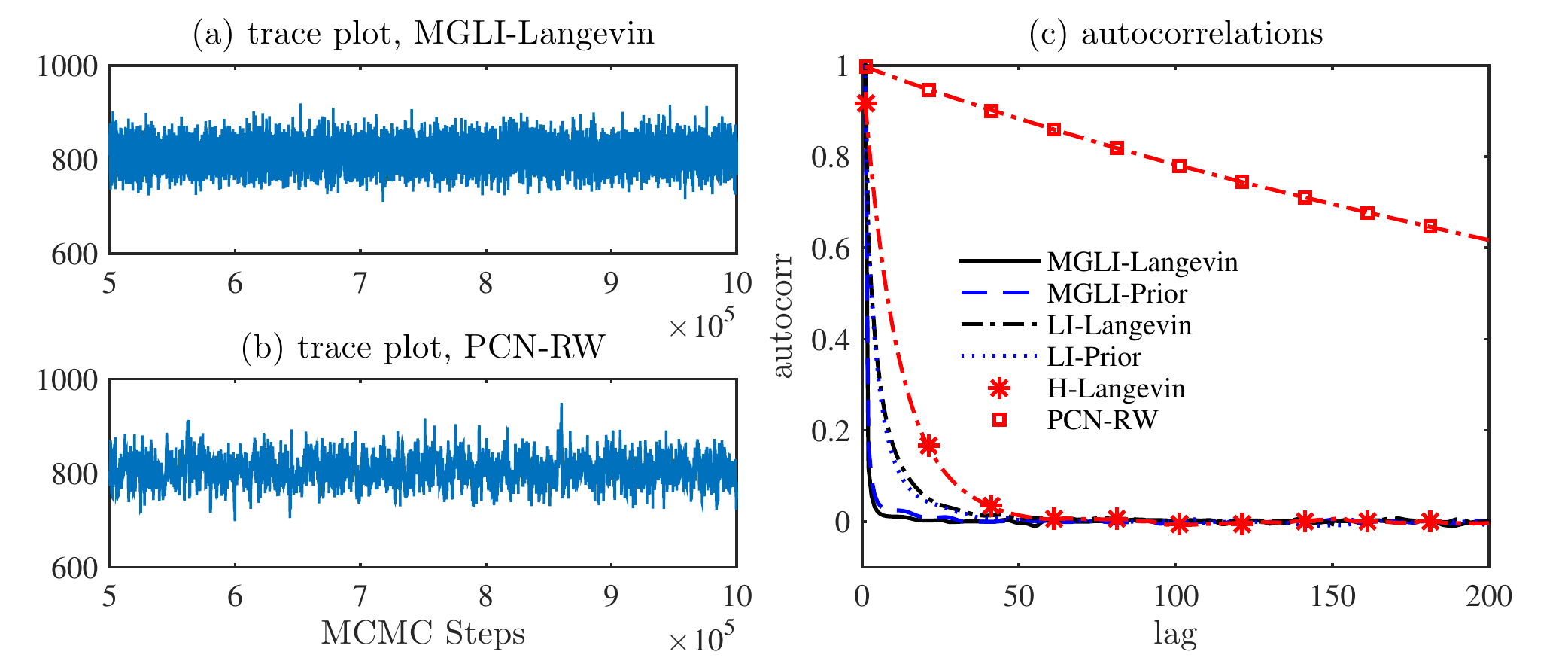}}
\caption{Summaries of MCMC mixing via the Onsager-Machlup functional, SNR10 case. (a) Trace plot of OMF values produced by MCMC with the MGLI-Langevin proposal. (b) Trace plot of OMF values produced by MCMC with the pCN-RW proposal. (c) Autocorrelation of OMF values produced by various proposals.}
\label{fig:SNR10_omf}
\end{figure}

\begin{figure}[h!]
\centerline{\includegraphics[width=0.75\textwidth]{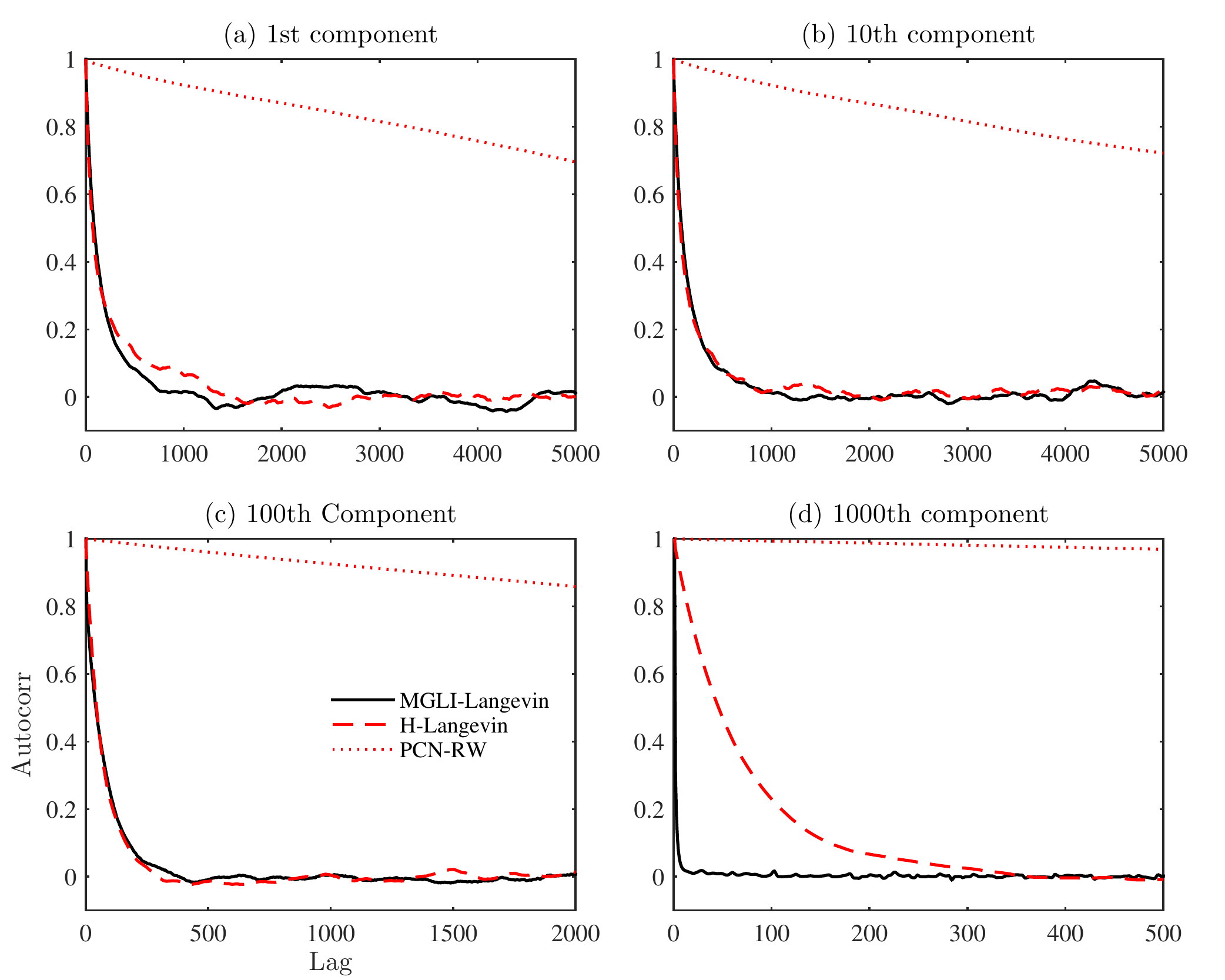}}
\caption{Autocorrelation functions of selected components of the $v$ parameter produced by various proposals, SNR50 case.}
\label{fig:SNR50_v}
\end{figure}

\begin{figure}[h!]
\centerline{\includegraphics[width=0.75\textwidth]{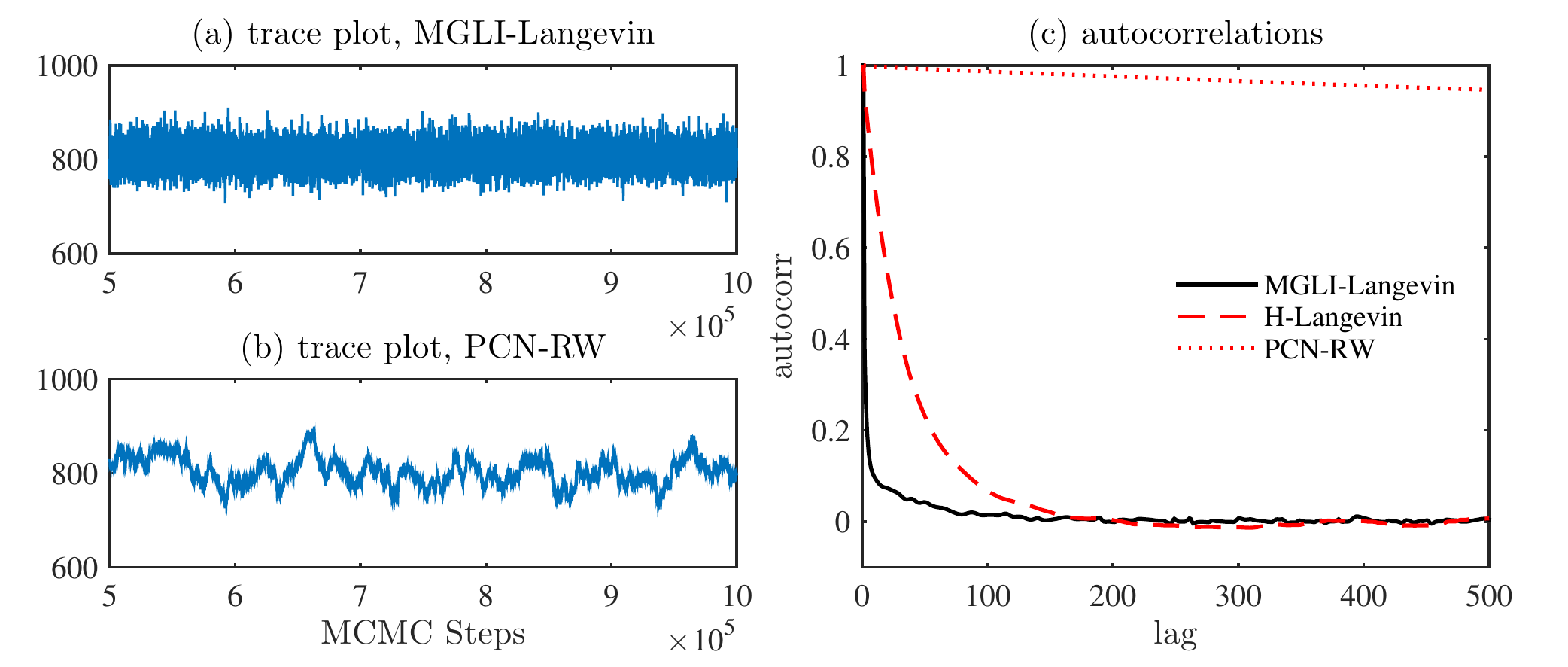}}
\caption{Summaries of MCMC mixing via the OMF, SNR50 case.}
\label{fig:SNR50_omf}
\end{figure}

In both of these test cases, H-Langevin and the operator-weighted proposals have comparable sampling efficiency along parameter directions that are dominated by the likelihood; this is not surprising, as both use Hessian information from the log-likelihood to scale proposed moves. For directions along which the likelihood has a relatively weaker influence than the prior, however, the operator-weighted proposals have better sampling efficiency than H-Langevin.
To further explore this pattern, Figure \ref{fig:auto_lag1} shows the lag-1 autocorrelation of all 1600 components of $v$, for both the SNR10 and  SNR50 cases (subfigures \ref{fig:auto_lag1}(a) and \ref{fig:auto_lag1}(b), respectively).
The star symbols represent the lag-1 autocorrelation of MCMC samples produced by H-Langevin while the square symbols represent results from MGLI-Langevin. 
We observe that the lag-1 autocorrelations produced by H-Langevin stay almost constant across components of $v$ (i.e., across the prior eigenfunctions).
In comparison, the lag-1 autocorrelations produced by MGLI-Langevin decrease significantly for the components of $v$ corresponding to the high-index eigenfunctions, and are lower on average than the autocorrelations from H-Langevin.
The other operator-weighted proposals produce similar results, though the lag-1 autocorrelations of the lowest-index eigenfunctions are not always smaller with the operator-weighed proposals than with H-Langevin. For brevity, these results are not shown. 

\begin{figure}[h!]
\centerline{\includegraphics[width=0.75\textwidth]{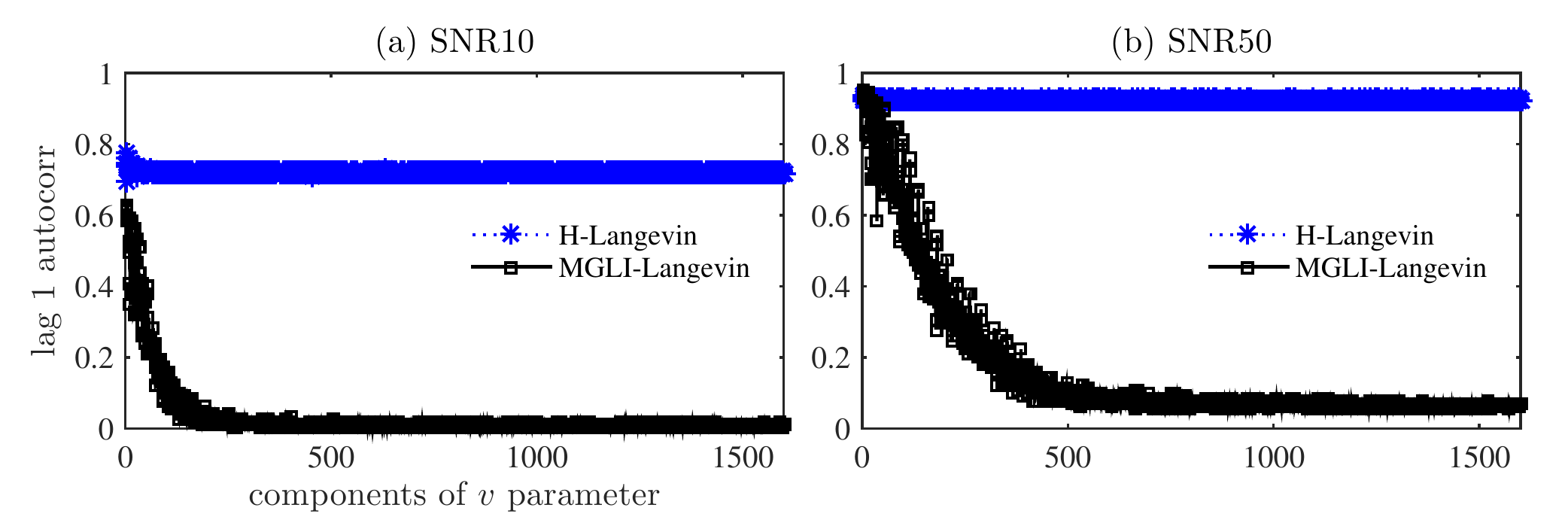}}
\caption{The lag-1 autocorrelation of each component of the $v$ parameter. Star and square symbols are the results of H-Langevin and MGLI-Langevin, respectively. (a) SNR10 case. (b) SNR50 case.}
\label{fig:auto_lag1}
\end{figure}

\added{
In the SNR10 case, the CPU times spent, after adaptive LIS construction, on $10^6$ iterations of MCMC with the LI-Prior, LI-Langevin, pCN-RW, and H-Langevin proposal are $0.88\times 10^5$ seconds, $1.12\times 10^5$ seconds, $0.85\times 10^5$ seconds and $1.25\times 10^5$ seconds, respectively. 
Here the CPU time of LI-Prior is comparable to that of pCN-RW, and the CPU time of LI-Langevin is comparable to that of H-Langevin. 
We also note that the Langevin proposals do not spend much additional CPU time on gradient evaluations compared to LI-Prior and pCN-RW, since in the present example the underlying PDE is self-adjoint and hence the forward solution can be recycled for the adjoint evaluation. 
The CPU times spent on $10^6$ MCMC steps with MGLI-Prior and MGLI-Langevin are $1.70 \times 10^5$ seconds and $2.10\times 10^5$ seconds, respectively. 
These essentially double the CPU times of LI-Prior and LI-Langevin, as expected. 
The CPU time spent on the adaptive construction of the LIS using $1000$ posterior samples is about $0.23 \times 10^5$ seconds, which is not significant relative to the time spent on MCMC sampling.
We note that using $1000$ samples here is a rather conservative choice. From our experience in \cite{Cui_LIS_2014} and the convergence study shown next (see Figure~\ref{fig:SNR50_conv}), about $500$ samples are sufficient to estimate an LIS with good accuracy.  
For the SNR50 case, the CPU times spent on each of the proposals and on the LIS construction are similar to those of the SNR10 case, and thus are not reported for brevity. 
}

\subsection{Global versus local LIS}

It is useful to quantify the impact of ``globalizing'' the LIS on sampling efficiency. In particular, we wish to contrast the performance of operator-weighted proposals constructed from a global LIS, based on the posterior expected Hessian (as in Section \ref{sec:lis}), with that of operator-weighted proposals constructed from a single local LIS, based on the Hessian at the MAP. 
For brevity, we limit our comparisons to Proposal~\ref{prop:mglis_mala} (MGLI-Langevin). 
The global LIS is constructed using the adaptive sampling strategy detailed at the start of Section~\ref{sec:samplingcompare}. Results produced with the global LIS are labeled `Adapt-LIS' in the figures below. Results obtained with the MGLI-Langevin proposal employing a local LIS at the MAP are denoted `MAP-LIS.' Note that the local LIS in the current problem has dimension at most 25, since the inverse problem has 25 observations. The global LIS can be much larger, since it accounts for posterior variation in the dominant eigenspace of the Hessian. In the current setup, the global LIS for the SNR10 case has dimension $r=66$ while for the SNR50 case it has dimension $r=193$; both of these are obtained with truncation thresholds $\tau_{\rm loc} = \tau_{\rm g} = 0.01$. 

Figures \ref{fig:SNR10_lis_map} and \ref{fig:SNR50_lis_map} summarize relative sampling performance for both the SNR10 and SNR50 cases. We show the autocorrelation of the OMF for each algorithm and the lag-1 autocorrelations of each component of $v$.
In both test cases, using a global LIS produces better mixing than a single local LIS. 
The improvement in mixing for the OMF seems slightly more pronounced for the larger-noise case (SNR10) than the smaller-noise case (SNR50). This may be due to the broader posterior of the former; since the nonlinear forward model is the same in both cases, the variation of the Hessian becomes more significant when wider ranges of the parameter space are explored. On the other hand, the lag-1 autocorrelations show greater improvements in the SNR50 case. Here, the higher dimension of the global LIS relative to the local LIS (193 to 25 in the SNR50 case versus 66 to 25 in the SNR10 case) may play a role. The global LIS is larger in the SNR50 case because the data carry more information.

\begin{figure}[h!]
\centerline{\includegraphics[width=0.75\textwidth]{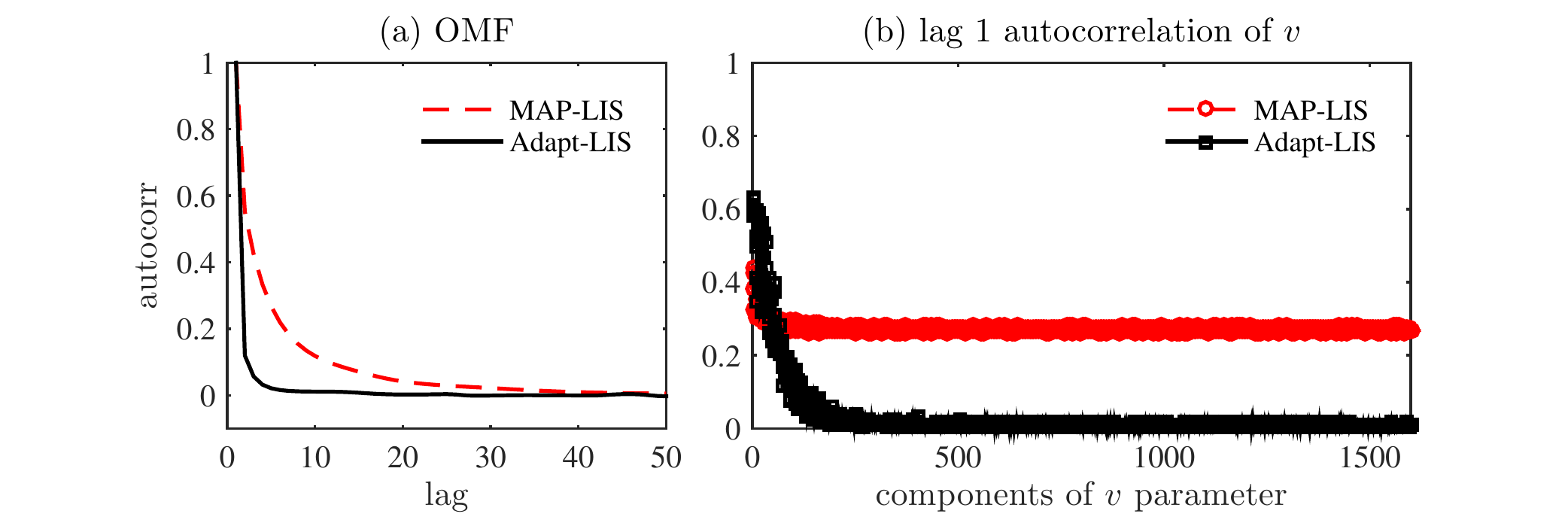}}
\caption{Sampling performance of Adapt-LIS (solid line) and MAP-LIS (dashed line), SNR10 case. (a) Autocorrelation of the OMF. (b) Lag-1 autocorrelation of each component of the $v$ parameter.}
\label{fig:SNR10_lis_map}
\end{figure}

\begin{figure}[h!]
\centerline{\includegraphics[width=0.75\textwidth]{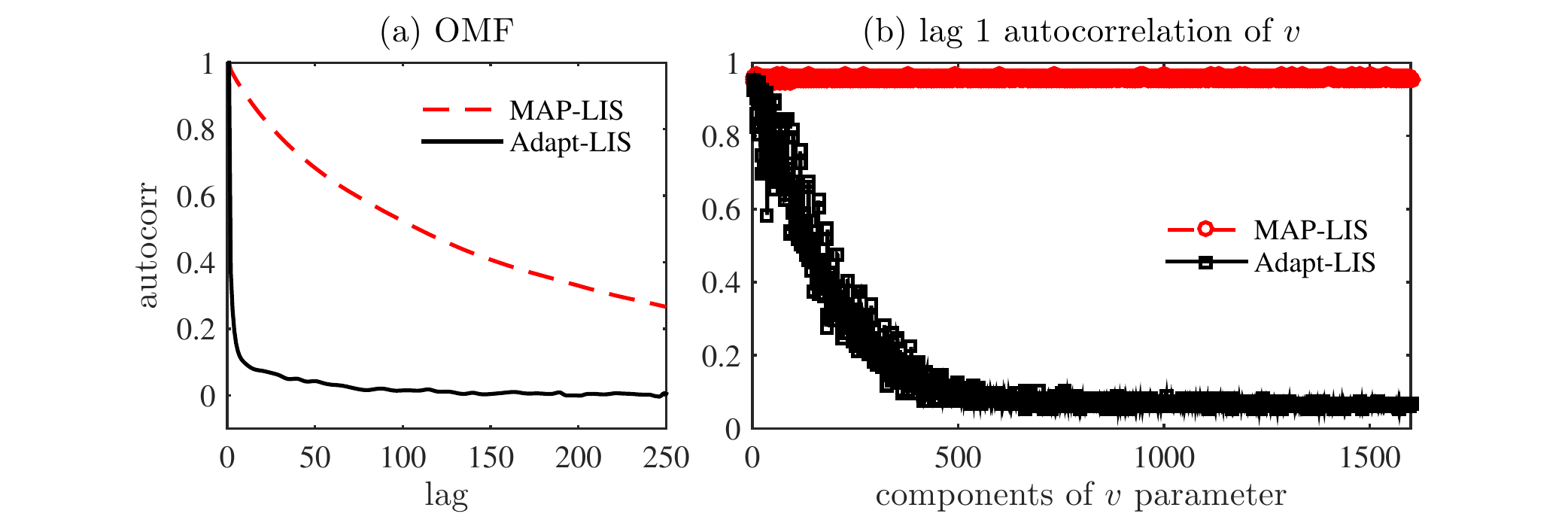}}
\caption{Sampling performance of Adapt-LIS (solid line) and MAP-LIS (dashed line), SNR50 case.}
\label{fig:SNR50_lis_map}
\end{figure}

\subsection{LIS under grid refinement}

We now explore how grid refinement affects the dimensionality and structure of the global LIS, as well as the convergence of the adaptive procedure for constructing it. In the examples below, the transmissivity field $\kappa(s)$ and the associated potential function $p(s)$ are discretized on $40 \times 40$, $80 \times 80$, and $120 \times 120$ grids.

Figure \ref{fig:SNR50_conv} shows the dimension of the LIS and the convergence diagnostic \eqref{eq:lis_conv} versus the number of samples used in the adaptive construction process, for the SNR50 case only.
\begin{figure}[h!]
\centerline{\includegraphics[width=0.75\textwidth]{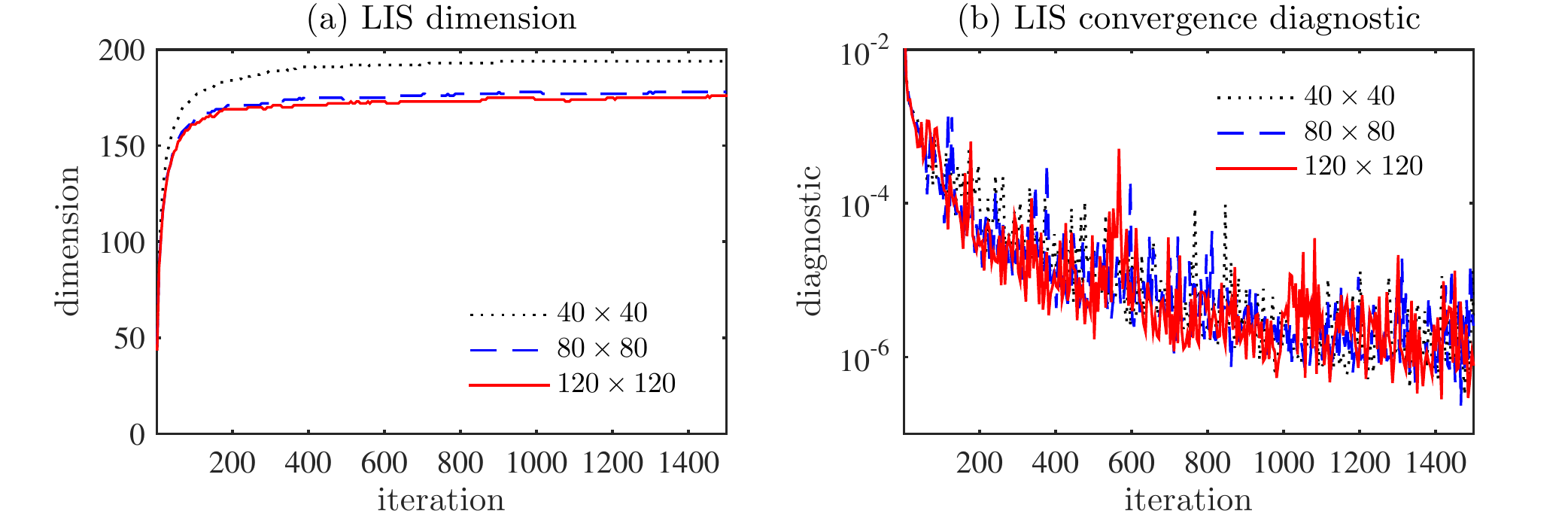}}
\caption{The dimension of the LIS and the subspace convergence diagnostic \eqref{eq:lis_conv} versus the number of samples used in the adaptive algorithm; SNR50 case. Dotted, dashed, and solid lines represent the $40\times 40$, $80\times 80$, and $120\times 120$ grids, respectively. (a) The dimension of the LIS. (b) The convergence diagnostic.}
\label{fig:SNR50_conv}
\end{figure}
For all three discretizations, the distance \eqref{eq:lis_conv} between likelihood-informed subspaces at adjacent iterations drops by several orders of magnitude over the course of the adaptive sampling procedure, as shown in Figure \ref{fig:SNR50_conv}(b). The rates of convergence of this distance are comparable for all three discretizations. 
After a few hundred samples, the dimensionality of LIS also converges for all three discretizations. Note that the $40\times40$ grid yields a slightly higher-dimensional LIS than the two more refined grids: at the end of the adaptive procedure, the LIS of the $40\times 40$ grid has dimension $r=193$, while the $80\times 80$ and $120\times 120$ grids yield $r=178$ and $r=176$, respectively. This effect can be ascribed to discretization errors on the $40\times 40$ grid; since the forward model converges under grid refinement, we expect the dimension of the associated LIS also to converge. The first five basis vectors of the LIS for the various discretizations are plotted in Figure \ref{fig:SNR50_basis}. Similar structures are observed at different levels of grid refinement.

\begin{figure}[h!]
\centerline{\includegraphics[width=0.75\textwidth]{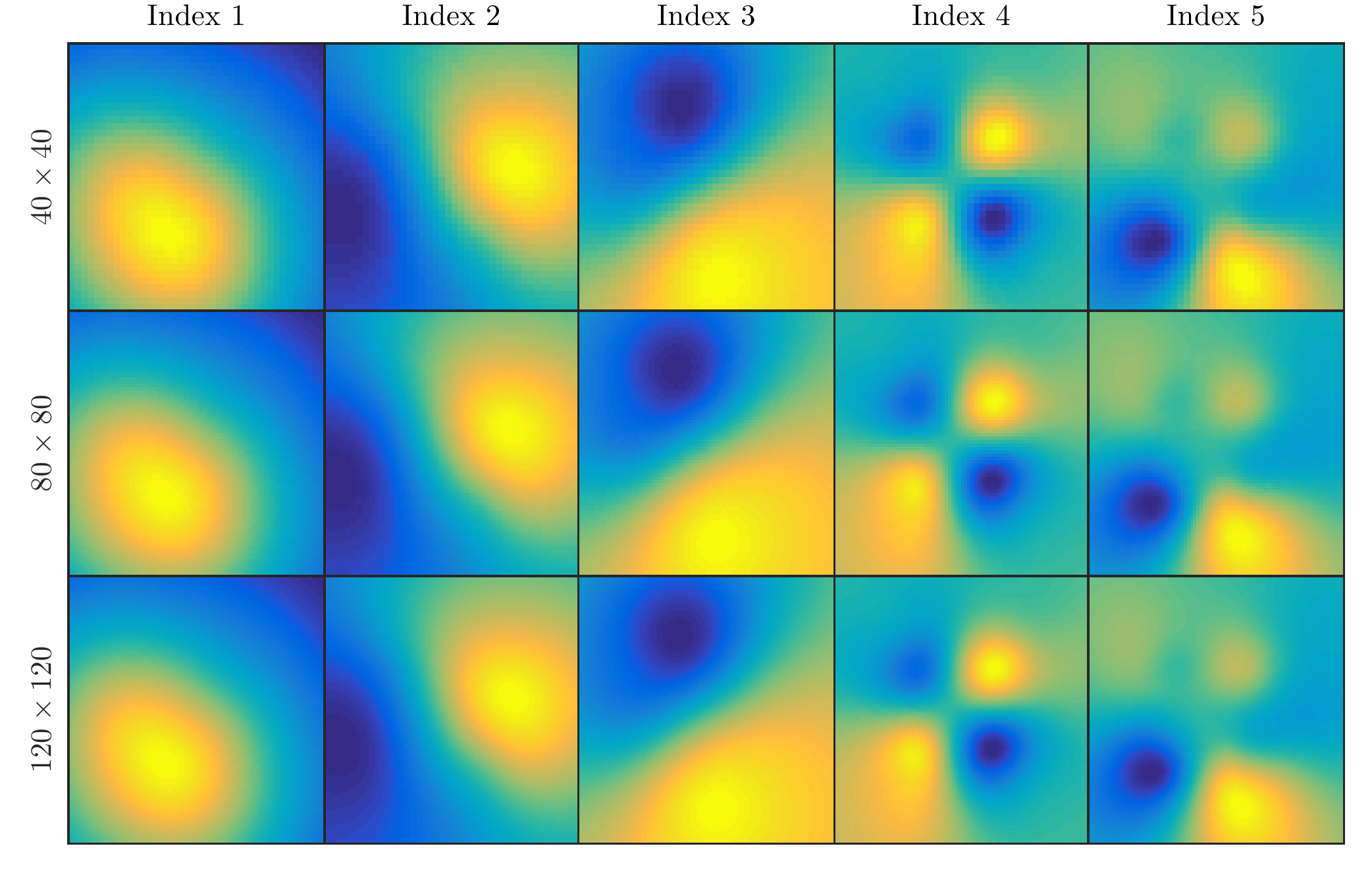}}
\caption{The first five basis functions of the LIS for three different levels of grid refinement; SNR50 case. Top row: $40\times 40$ grid, middle row: $80\times 80$ grid, and bottom row: $120\times 120$ grid.}
\label{fig:SNR50_basis}
\end{figure}

%% file: conditioned_diffusion_r1.tex

\section{Example 2: Conditioned diffusion}
\label{sec:cd}

In this section, we use noisy point-wise observations of the path of a particle, diffusing in a double-well potential, to infer the driving force on the particle and hence its pushforward to the path itself.  This example is motivated by applications in molecular dynamics.

\subsection{Problem setup}

Consider a stochastic process $p : [0,T] \rightarrow \R$ governed by the following Langevin SDE:
\begin{equation}
\label{eq:langevin1d}
dp_t = f(p_t)dt + du_t, \quad p_0 = 0,
\end{equation}
where $f: \R \rightarrow \R$ is globally Lipschitz and $du_t$ is an increment of the Brownian motion $u \sim \mu_0 = \normal(0,C)$, where $C(t,t')= \min (t,t')$. 
Let the function $f$ have the following form 
\begin{equation*}
f( p ) := \beta p (1-p^2)/(1+p^2),
\end{equation*}
with $\beta>0$.  The corresponding potential is $E(p) = -\int_{-\infty}^{p} f(s) ds$. From any initial condition $p_0$, the state $p_t$ approaches the invariant measure of the SDE, which has a density $Z^{-1} \exp\left (-2E(p) \right )$ with respect to Lebesgue measure, where $Z = \int_{\R} \exp\left (-2E(s) \right ) ds$. 
Since the potential $E$ has a double-well shape, the invariant measure is bimodal and paths of \eqref{eq:langevin1d} will transition from one well to the other, with a probability dependent on the magnitude of the stochastic forcing and on the scale $\beta$ of the potential. We will use $\beta =10$ in the examples below.
This model is ubiquitous in the sciences, perhaps most notably in molecular dynamics where it represents the motion of a particle with negligible mass trapped in an energy potential $E$ with thermal fluctuations represented by the Brownian forcing.

We refer the reader to \cite{HSV_2011} for a proof that the map $u \mapsto p$ is continuous and differentiable from $C([0,T],\R) \rightarrow C([0,T],\R)$. The fact that $\mu_0(C([0,T],\R)) = 1$ follows from the well-known property that continuous functions have probability one under Wiener measure. Therefore, this Bayesian inverse problem is well-defined \cite{HSV_2011}, and it fits into the framework developed in Section \ref{sec:lkd_informed}. 
The observation operator is defined by 
$\cM p := [p_{t_1},p_{t_2},\dots,p_{t_{20}}]^T$, and we let
$$
y = \cM p + e, \quad e \sim  \normal(0,\sigma^2 I_{20}),
$$
where $\sigma=0.1$ and the observation times $t_i$ are equispaced within the interval $[0,10]$. An Euler-Maruyama scheme is used for integration with $\Delta t=10^{-2}$, and hence the dimension of our approximation of this infinite-dimensional pathspace is $N=1000$. 

\begin{figure}[h!]
\centerline{\includegraphics[width=0.75\textwidth]{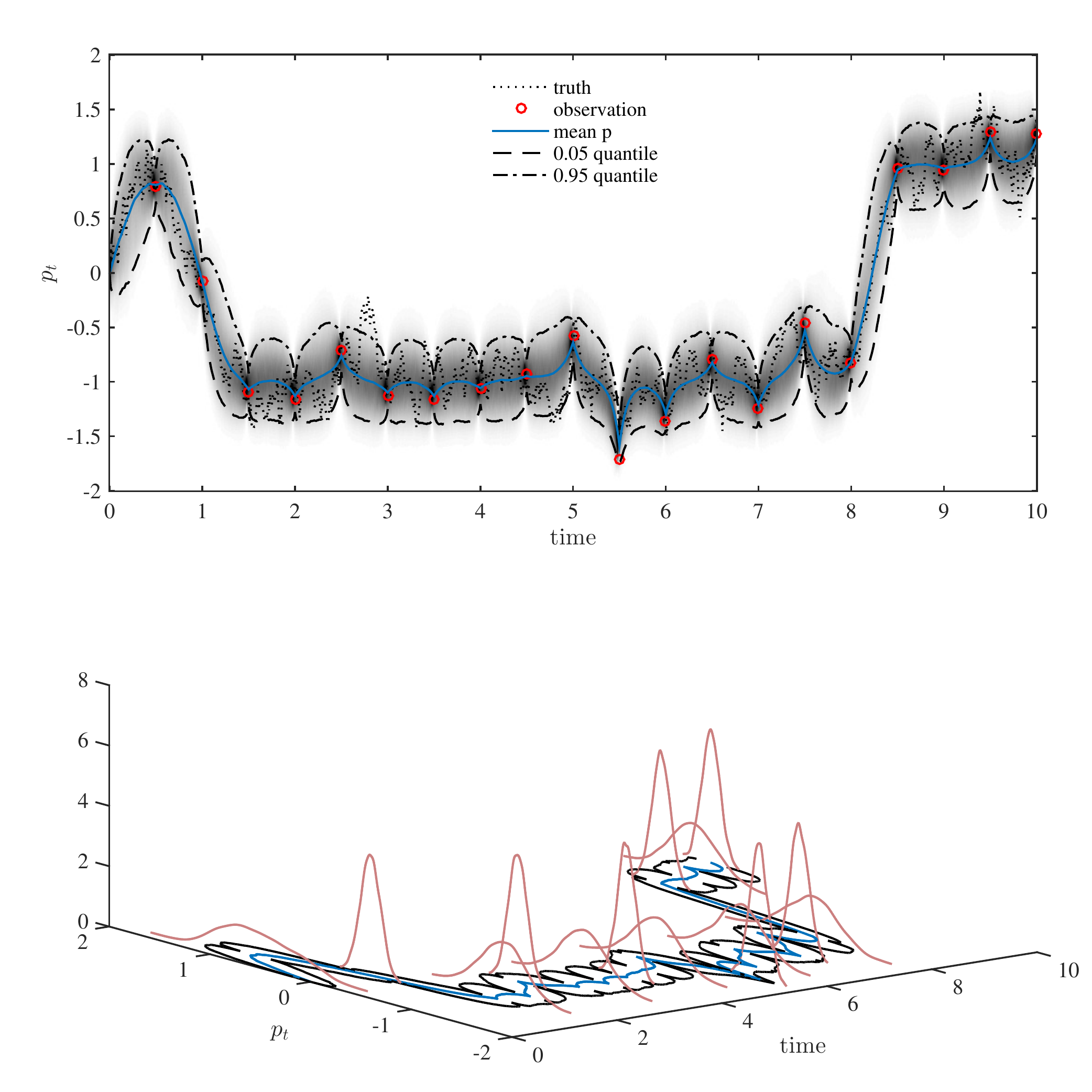}}
\caption{Setup of the conditioned diffusion example. Top: the true trajectory, the noisy observations, the posterior mean of the trajectory, the $0.05$ and $0.95$ quantiles of the trajectory marginalized at each time. The posterior density of the trajectory marginalized at each time is shown as the grayscale background image.
Bottom: the posterior density of the trajectory marginalized at selected times.}
\label{fig:cd_setup}
\end{figure}

\begin{figure}[h!]
\centerline{\includegraphics[width=0.75\textwidth]{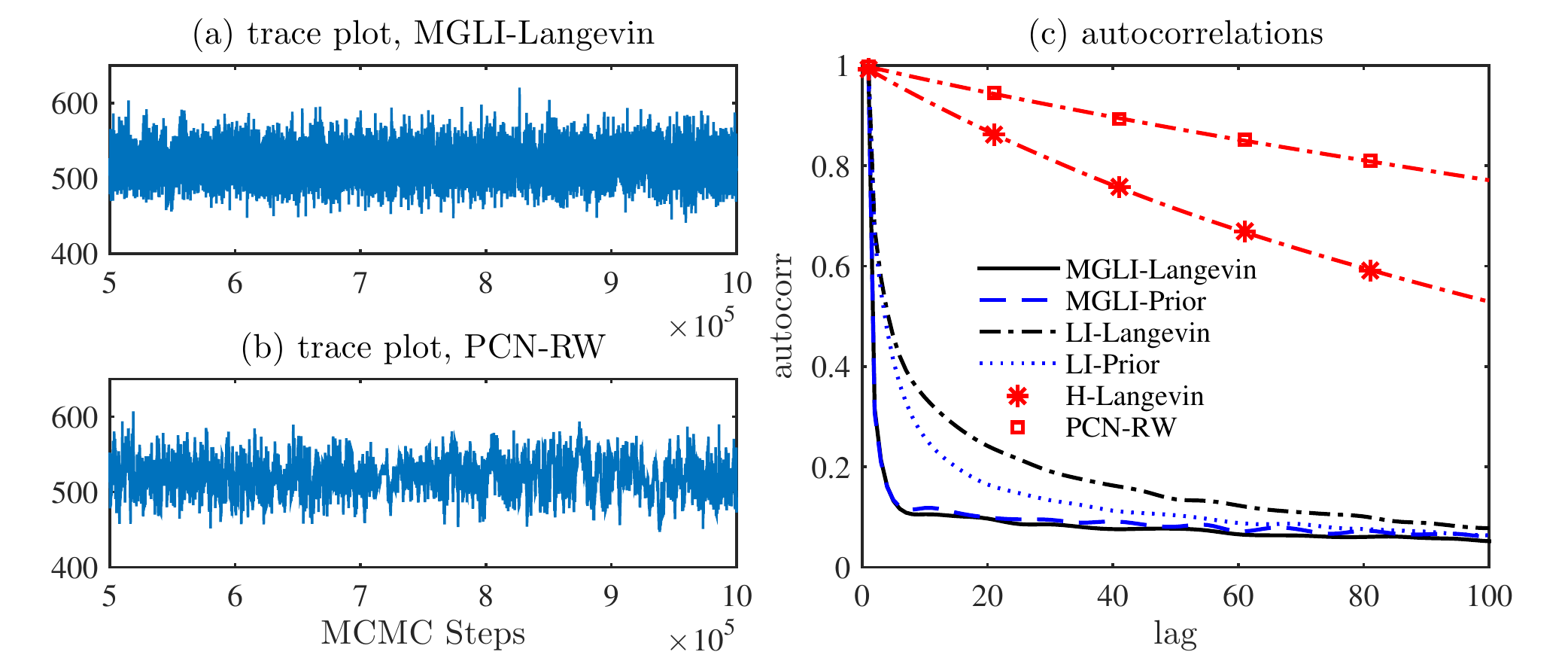}}
\caption{Summaries of MCMC mixing for the conditioned diffusion example. (a) Trace plot of the OMF produced by MGLI-Langevin. (b) Trace plot of the OMF produced by pCN-RW. (c) Autocorrelation of the OMF, as produced by various proposals.}
\label{fig:cd_omf}
\end{figure}
 
\begin{figure}[h!]
\centerline{\includegraphics[width=0.4\textwidth]{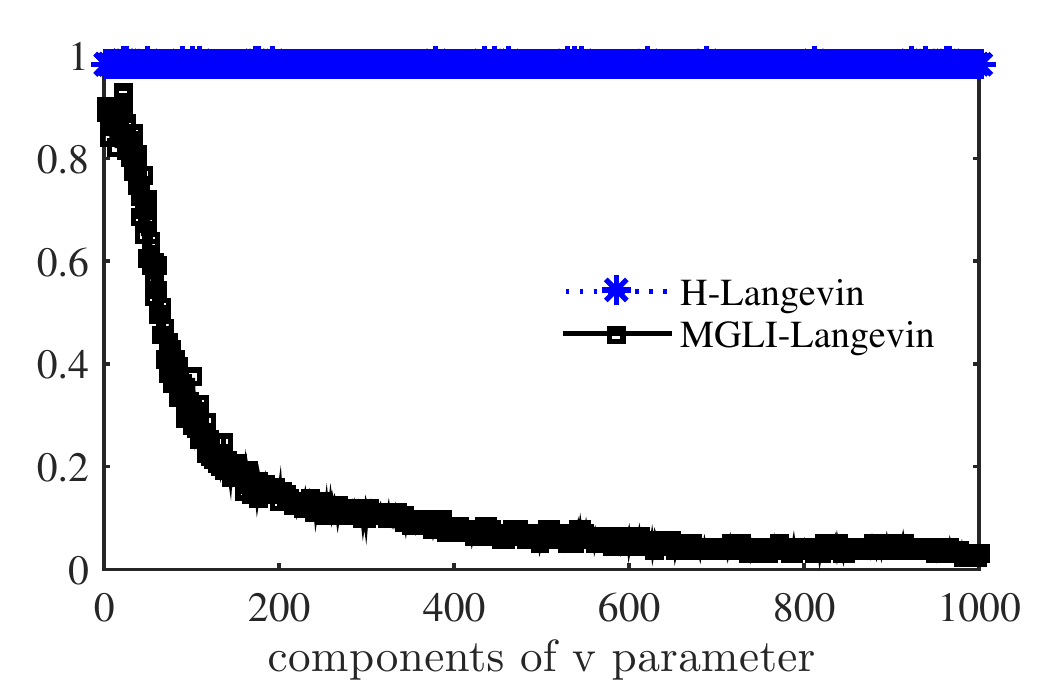}}
\caption{The lag-1 autocorrelation of each component of the $v$ parameter for the conditioned diffusion example. Star and square symbols are produced by H-Langevin and MGLI-Langevin proposals, respectively. }
\label{fig:cd_auto_lag1}
\end{figure}

\begin{figure}[h!]
\centerline{\includegraphics[width=0.75\textwidth]{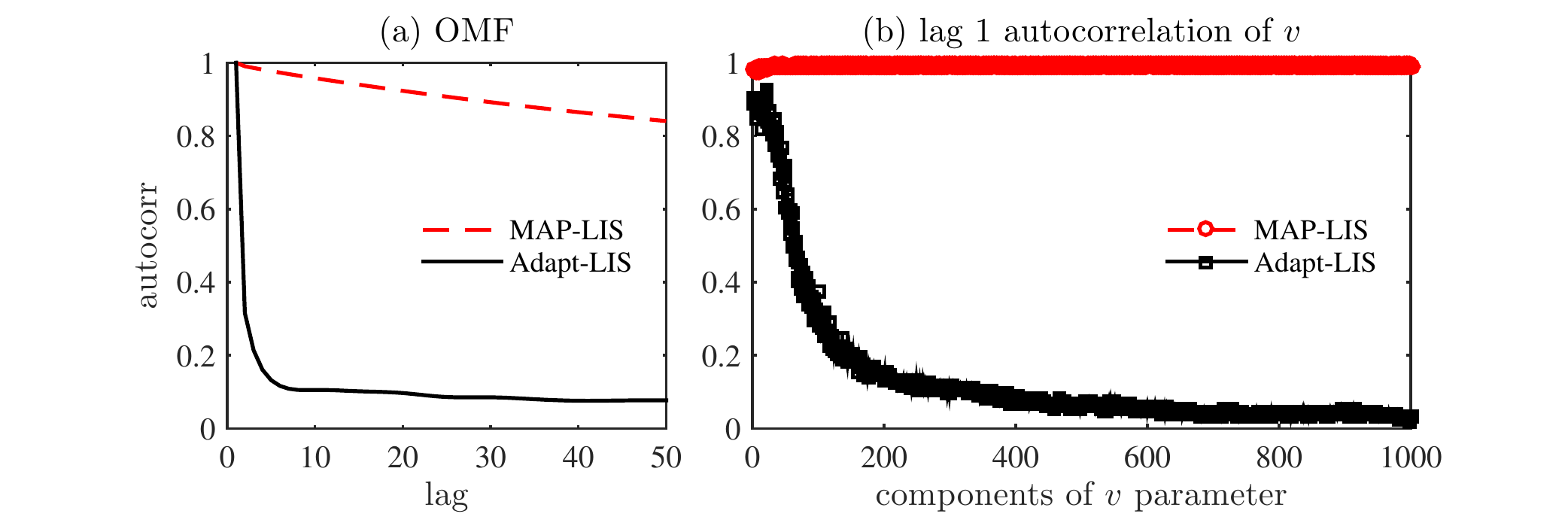}}
\caption{Sampling performance of MGLI-Langevin with a global LIS (Adapt-LIS, solid line) versus a local LIS (MAP-LIS, dashed line) for the conditioned diffusion example. (a) Autocorrelation of the OMF. (b) Lag-1 autocorrelation of each component of the $v$ parameter.}
\label{fig:cd_lis_map}
\end{figure}

\subsection{Numerical results}

We now evaluate the efficiency of the MCMC proposals from Section~\ref{sec:lis_operators} and Section \ref{sec:bench} on the conditioned diffusion example. Figure \ref{fig:cd_setup} illustrates the solution to the forward problem taken as the truth, the noisy observations of this particular path, and the results of subsequent conditioning---i.e., aspects of the posterior. In particular, we show the posterior mean, marginals of the posterior density at each time step, and the $0.05$ and $0.95$ quantiles of the local marginals. 
Note that the time interval $t \in [0,10]$ considered here is long enough to capture two transitions between the potential wells. Our numerical experiments use the same chain lengths and burn-in intervals as in the elliptic example, detailed in Section~\ref{sec:samplingcompare}. We also use the same number of iterations as before to adaptively construct the global LIS.

The numerical results in Figures~\ref{fig:cd_omf}--\ref{fig:cd_lis_map} highlight the strong performance of the DILI samplers developed in this paper. Figure~\ref{fig:cd_omf} summarizes the sampling efficiency of each MCMC scheme via traces and autocorrelations of the OMF. Not only do the DILI algorithms outperform pCN, but they also show dramatically improved mixing over H-Langevin; here, H-Langevin performs almost as poorly as pCN. This contrasts with the elliptic PDE cases examined in Section \ref{sec:elliptic}, where the performance of H-Langevin was reasonably close to that of the DILI proposals. Plots of lag-1 autocorrelation in Figure \ref{fig:cd_auto_lag1} reinforce this notion, showing much faster decay of lag-1 autocorrelation with MGLI-Langevin than with H-Langevin across the first 1000 eigenfunctions of the prior. 

Figure~\ref{fig:cd_lis_map} shows how mixing performance is improved by using a global LIS rather than a local LIS based on the Hessian at the MAP. Both traces in this figure use the MGLI-Langevin algorithm, but those marked `Adapt-LIS' use a global LIS constructed via the adaptive algorithm of Section~\ref{sec:adaptsamp}. The global LIS clearly yields a much faster decay of the OMF autocorrelation and better mixing, as measured by lag-1 autocorrelations, across all components of $v$.

%% file: conclusion_r1.tex

\section{Conclusions}
\label{sec:conc}

This paper has introduced a general class of operator-weighted MCMC proposals that are well defined on function space. In Bayesian inference problems that require exploring a posterior distribution over functions, these proposals yield MCMC sampling performance that is invariant under refinement of the discretization of these functions, and hence dimension-independent. 
While this class includes earlier dimension-independent algorithms \cite{BRSV_2008, CRSW_2012, Proposal_Law_2012} as particular cases, we use the versatility offered by our new family of proposals to design samplers that exploit posterior structure wherein departures from the prior---including non-Gaussianity---are concentrated on a finite number of directions, captured by a global likelihood-informed subspace (LIS). The global LIS is identified by approximating the posterior expectation of the Hessian of a preconditioned data-misfit functional. A further adaptive strategy based on this decomposition yields computationally affordable approximations of the posterior covariance. All of this geometric information is used to construct \textit{dimension-independent} and \textit{likelihood-informed} (DILI) proposal distributions; the four variations studied here (Proposals~\ref{prop:lis_prior}--\ref{prop:mglis_mala}) also make use of local gradient information and Metropolis-within-Gibbs updates. Numerical experiments for two nonlinear inverse problems suggest that these DILI samplers offer significant gains in sampling efficiency over current state-of-the-art algorithms.

This work can be extended in many directions. First, one can certainly consider alternative constructions of the LIS. Of particular interest are parallel methods for constructing the global LIS, strategies for infinite batch-updating of the LIS that are compatible with adaptive MCMC, and even directed strategies for optimally exploring the variation of the preconditioned data-misfit Hessian. 
If only gradients of the data-misfit function are available, one might resort to so-called active subspace methods \cite{Russi_2010, active_subspace_2014} as an alternative way of building up LIS-like information.  If gradients are also not available, then a different approach might be taken, for example adaptively building covariance information from samples and using regularized estimates to build an approximate LIS.

Second, many other operator-weighted proposal constructions are possible. For example, given the approximation of the 
empirical posterior covariance $\Sigma = \basis_r^{} \left( \cD_r^{} - \cI_r^{} \right) \basis_r^{\ast} + \cI$ defined in \eqref{eq:low_rank_cov}, we can define the operators 
\begin{eqnarray*}
\cA & = & (a_r^{} - a_\perp^{}) \basis_r^{} \basis_r^{\ast} + a_\perp^{} \cI \\
\cB & = &  \basis_r^{} ( b_r^{} {\cD_r^{}}^{\half} - b_\perp^{} \cI_r^{} ) \basis_r^{\ast} + b_\perp^{} \cI \\
\cG & = & 0,
\end{eqnarray*}
where $a_r, a_\perp \in (-1,1)$, and $b_r$ and $b_\perp$ take values $\sqrt{1 - a_r^2}$ and $\sqrt{1 - a_\perp^2}$, respectively.
Given a set of posterior samples $\{u_1, \ldots, u_n\}$, by setting the reference variable $m_{\rm ref}$ to be the difference of the empirical posterior mean and the prior mean, i.e., $m_{\rm ref} = \frac{1}{n}\sum_{k=1}^n u_k - m_0$, the resulting proposal \eqref{eq:operator_weighted} has a Gaussian approximation of the posterior as its invariant distribution.  

Finally, while this work has focused on the design of global operators (e.g., $\cA$, $\cB$, and $\cG$ that do not depend on the local parameter value), there is much room to combine global and local operators. As an example of the latter, \cite{beskos2014stable} uses a semi-implicit discretization of a locally-preconditioned Langevin equation to derive a proposal that, under appropriate assumptions, yields a dimension-independent MCMC algorithm. The present work may provide a rather general way of constructing local preconditioners that satisfy these assumptions. In particular, the partition of parameter space induced by the LIS provides an opportunity to introduce a variety of efficient proposals on the LIS (e.g., locally-preconditioned manifold MALA, even RMHMC) while preserving dimension independence through a suitable discretization of a Langevin equation on the CS. Combining proposals in this manner extends beyond the notion of local preconditioning. We leave a fuller development of these ideas to future work.

%% file: appendix_r1.tex

\appendix

\section{Feldman-Hajek Theorem}
\label{sec:FH_theorem}
By the Feldman-Hajek Theorem (see Theorem 2.23 of \cite{DaPrato_1992}), a pair of Gaussian measures $\normal( m_1 , \Gamma_1 )$ and $\normal( m_2 , \Gamma_2 )$ are equivalent measures if and only if the following conditions are satisfied: 
\begin{enumerate}

\item The two measures have the same  Cameron-Martin space, i.e., $\im \left (\Gamma_1^{\half} \right ) = \im \left (\Gamma_2^{\half} \right)$.

\item The difference between the mean functions is in the Cameron-Martin space, i.e., $ m_1 - m_2 \in {\rm Im}(\Gamma_1^{\half})$.

\item The operator 
\(
\cT = \left(\Gamma_1^{-\half} \Gamma_2^{\half} \right)
\left(\Gamma_1^{-\half} \Gamma_2^{\half}\right)^{\ast} - \cI 
\)
is Hilbert-Schmidt.
\end{enumerate}

\section{Proof of Theorem \ref{theo:1}}
\label{sec:proof:t1}
Let $q(u, du')$ denote the proposal distribution defined by \eqref{eq:operator_weighted}, which has the form
\begin{eqnarray}
u' & = & \prcov^{\half}\cA\prcov^{-\half} (u - u_{\rm ref}) + u_{\rm ref} + \prcov^{\half}\cB \rand - \prcov^{\half} \cG \prcov^{\half} \gradientu \nonumber \\
 & = & \prcov^{\half}\cA\prcov^{-\half} (u - u_{\rm ref}) + u_{\rm ref} + \prcov^{\half} \basis \left( \cD_\cB \rand' - \cD_\cG \basis^{\ast} \prcov^{\half} \gradientu \right)
\end{eqnarray}
where $\rand \sim \normal(0, \cI)$, $\rand' = \basis^{\ast} \rand$, and $u_{\rm ref} = m_0 + m_{\rm ref}$ with $m_{\rm ref} \in \im(\prcov^{\half})$. We note that $\rand' \sim \normal(0, \cI)$ because $\basis$ is an unitary operator.

The term $\cD_\cB \rand' - \cD_\cG \basis^{\ast} \prcov^{\half} \gradientu$ can be decomposed as
\[
\left\{ \begin{array}{ll} b_i \rand_i' - g_i {u_i}^{\ast} \left( \prcov^{\half} \gradientu \right) & {\rm if} \; b_i \neq 0\\ 0 & {\rm if} \; b_i = 0\end{array} \right.
\]
Therefore the Gaussian laws for $\prcov^{\half} \basis \cD_\cB \rand' $ and
$\prcov^{\half} \basis \left( \cD_\cB \rand' - \cD_\cG \basis^{\ast} \prcov^{\half}
  \gradientu \right)$ are equivalent since, by assumption, $\prcov
  \gradientu \in \im(\prcov^{\half})$.

Following the argument in the proof of Theorem 4.1 of \cite{BRSV_2008}, we can define a simplified proposal distribution $\tilde{q}(u, du')$ by
\begin{eqnarray}
u' & = & \prcov^{\half}\cA\prcov^{-\half} (u - u_{\rm ref}) + u_{\rm ref} + \prcov^{\half} \basis \cD_\cB \rand' \nonumber \\
 & = & \prcov^{\half}\cA\prcov^{-\half} (u - u_{\rm ref}) + u_{\rm ref} + \prcov^{\half}\cB \rand,
\end{eqnarray}
where $\tilde{q}(u, du')$ is equivalent to $q(u, du')$.
Now we consider the pair of measures 
\begin{eqnarray}
\tilde{\nu}(du, du') & = & \tilde{q}(u, du')\mu_0(du) , \nonumber \\
\tilde{\nu}^\bot(du, du') & = & \tilde{q}(u', du) \mu_0(du'),
\end{eqnarray}
which are joint Gaussian measures on $(u, u')$.
Under the assumption that $\mu_y$ and $\mu_0$ are equivalent, $\tilde{\nu}$ and $\tilde{\nu}^\bot$ are equivalent to the measures ${\nu}(du, du') = q(u, du')\mu_y(du)$ and 
${\nu}^\bot(du, du') = q(u', du) \mu_y(du')$, respectively.
To provide a well-defined MH algorithm for the infinite dimensional posterior measure $\mu_y$, one requires that $\tilde{\nu}$ and $\tilde{\nu}^\bot$ are equivalent measures, i.e., that
\[
\frac{d\tilde{\nu}^\bot}{d\tilde{\nu}}(u, u') 
\]
is positive and bounded $\tilde{\nu}$-almost surely. 
In the rest of this proof, we will establish the equivalence of $\tilde{\nu}$ and $\tilde{\nu}^\bot$, given the conditions provided above.

The pair of Gaussian measures $\tilde{\nu}$ and $\tilde{\nu}^\bot$ have the form 
\begin{equation}
\label{eq:ref_measure}
\tilde{\nu}(du, du') = \normal\left( m_1 , \cC_1 \right) \; {\rm and} \;
\tilde{\nu}^\bot(du, du') = \normal\left( m_2 , \cC_2 \right).
\end{equation}
Recall that $u_{\rm ref} = m_0 + m_{\rm ref}$, where $m_{\rm ref} \in \im (\prcov^{\half})$. The mean functions of $\tilde{\nu}$ and $\tilde{\nu}^\bot$ have the form
\begin{equation}
\label{eq:mean_function}
m_1 = \left( \begin{array}{l} m_0 \\ m_0 + m_{\rm ref} - \prcov^{\half}\cA\prcov^{-\half} m_{\rm ref} \end{array} \right)
\; {\rm and} \;
m_2 = \left( \begin{array}{l} m_0 + m_{\rm ref} - \prcov^{\half}\cA\prcov^{-\half} m_{\rm ref} \\ m_0 \end{array} \right),
\end{equation} 
respectively. 
The covariance operators of $\tilde{\nu}$ and $\tilde{\nu}^\bot$ are
\begin{equation}
\label{eq:cov_operator}
\cC_1 = \cL \cQ_1 \cL^{\ast} \; {\rm and} \; \cC_2 = \cL \cQ_2 \cL^{\ast},
\end{equation}
respectively, where
\begin{equation*}
\cQ_1 =
\left( \begin{array}{cc}
\cI & \cD_\cA \\
\cD_\cA & \cD_\cA^2 + \cD_\cB^2
\end{array} \right), \;
\cQ_2 = 
\left( \begin{array}{cc}
\cD_\cA^2 + \cD_\cB^2 & \cD_\cA \\
\cD_\cA & \cI
\end{array} \right), \;{\rm and} \;
\cL = 
\left( \begin{array}{cc} \prcov^{\half}\basis & 0 \\ 0 & \prcov^{\half}\basis \end{array} \right).
\end{equation*}

We now apply the Feldman-Hajek (FH) Theorem (see \ref{sec:FH_theorem}) to
show that $\tilde{\nu}$ and $\tilde{\nu}^\bot$ are equivalent
measures, given that the conditions in Theorem \ref{theo:1} hold. 
First we consider the mean (condition (2) of the FH theorem).
The difference between the mean functions $m_1$ and $m_2$ is 
\[
m_1 - m_2 = \left( \begin{array}{l} \prcov^{\half}\cA\prcov^{-\half} m_{\rm ref} - m_{\rm ref} \\ m_{\rm ref} - \prcov^{\half}\cA\prcov^{-\half} m_{\rm ref} \end{array} \right),
\] 
by \eqref{eq:mean_function}.
Since $m_{\rm ref} \in \im(\prcov^{\half})$ and $\cA$ is a bounded operator, we have $m_1 - m_2 \in \im(\cC_1^{\half})$.

To show the covariance conditions of the FH theorem, first we decompose the space into 
$\hilbert' \oplus \hilbert_0$ where
$\hilbert' = {\rm span} \{\psi_i\}_{b_i=0}$ and 
$\hilbert_0 = {\rm span} \{\psi_i\}_{b_i\neq0}$.
Furthermore, note that due to the block structure of diagonal
operators, we decompose the product space into $\hilbert^{\prime 2} \oplus \hilbert_0^2$, where $\cC_i=\cC'_i \oplus \cC_{0,i}$ and $\cC'_i$ and $\cC_{0,i}$ act only 
on $\hilbert'^2$ and $\hilbert_0^2$ respectively.
Note that $\cC'_1 = \cC'_2$ by Assumption\ref{assum1} (2), so conditions (1) and (3) of the FH theorem are automatically satisfied.  
We will consider the measures restricted to the $\hilbert^2_0$ space.  Indeed, without loss of generality, we will assume $|b_i| \geq c>0$ for all $i\in \N$ for the remainder of the proof, so that 
$\hilbert^2=\hilbert_0^2$.

For condition (1) of the FH theorem, note that given conditions (1)--(3) of Theorem \ref{theo:1}, there exist constants $c_1, c_2 > 0$ such that
\[
c_1 \|x \| < \|\cQ_1 x \| < c_2 \|x \| \; {\rm and} \; c_1 \|x \| < \| \cQ_2 x \| < c_2 \|x \|, \; \forall x \in \hilbert^2. %
\]
Thus the operators $\cQ_1^{\half}$ and $\cQ_2^{\half}$ have the same image, i.e., $\im(\cQ_2^{\half}) = \im(\cQ_2^{\half})$.  
Since $\cC_i^{\half}=\cL \cQ_i^{\half}$ for $i=1,2$, 
we have immediately $\im(\cC_1^{\half}) = \im(\cC_2^{\half})$.

For condition (3) of the FH theorem, note that the operators $\cQ_1$ and $\cQ_2$ can be decomposed as 
$\cQ_1^{} = \cR_1^{} \cR_1^{\ast}$ and $\cQ_2^{} = \cR_2^{} \cR_2^\ast$, where
\begin{equation}
\label{eq:QR}
\cR_1 = 
\left( \begin{array}{cc}
\cI & 0\\
\cD_\cA & \cD_\cB
\end{array} \right) \; {\rm and} \;
\cR_2 = 
\left( \begin{array}{cc}
\cD_\cB & \cD_\cA \\
0 & \cI
\end{array} \right).
\end{equation}
Thus $\cC_1$ and $\cC_2$ can be decomposed as 
\[
\cC_1 = \left( \cL \cR_1 \right)\left( \cL \cR_1 \right)^{\ast} \;{\rm and}\; \cC_2 = \left( \cL \cR_2 \right)\left( \cL \cR_2 \right)^{\ast},
\]
respectively.
Recollecting terms we have 
\begin{equation}
\cT = 
\left( \begin{array}{ll}
\cD_\cA^2 + \cD_\cB^2 - \cI & - \cD_\cA \cD_\cB^{-1} \left( \cD_\cA^2 + \cD_\cB^2 - \cI \right) \\
- \cD_\cA \cD_\cB^{-1} \left( \cD_\cA^2 + \cD_\cB^2 - \cI \right) & \cD_\cA^2 + \cD_\cB^{-2}\left( \cI - \cD_\cA^2 \right)^2 - \cI
\end{array} \right).
\end{equation}
So
\[
{\rm trace} \left( \cT^{\ast} \cT  \right) = \sum_{i = 1}^{\infty} \frac{\left( a_i^2 + b_i^2 - 1 \right)^2 \left( \left( a_i^2 + b_i^2 - 1 \right)^2 + 2b_i^2 \right)}{b_i^4} < K \sum_{i = 1}^{\infty} \left( a_i^2 + b_i^2 - 1 \right)^2,
\]
for some $K<\infty$, and hence the operator $\cT$ is Hilbert-Schmidt given that condition (4) of Theorem \ref{theo:1} holds. 

Therefore for operators $\cA$, $\cB$, and $\cG$ that satisfy the conditions of Theorem \ref{theo:1}, the proposal \eqref{eq:operator_weighted} yields a well-defined MH algorithm for the infinite dimensional posterior measure $\mu_y$ that is equivalent to the prior measure $\mu_0$.

\section{Proof of Corollary \ref{coro:1}}
\label{sec:proof:c1}

Condition (2) of Theorem \ref{theo:1} defines a conditional update, in which the components corresponding to $b_i = 0$ are fixed.
These fixed components do not affect the acceptance probability, so we will assume $|b_i| \geq c>0$ for all $i\in \N$ in this proof.

As in \ref{sec:proof:t1}, we first define the pair of proposals
\begin{eqnarray*}
q(v, dv') & : & v' = \cA v + \cB \rand - \cG \gradientv ,
\\
\tilde{q}(v, dv') & : & v' = \cA v + \cB \rand .
\end{eqnarray*}
Then we define the pair of measures $\nu(dv, dv')$ and $\tilde{\nu}(dv, dv')$ associated with this pair of proposals as $\nu(dv, dv') = q(v, dv') \mu_y (dv)$, and $\tilde{\nu}(dv, dv') = \tilde{q}(v, dv') \mu_0 (dv)$, respectively. 
Here $\mu_0 = \normal(0,\cI)$.
We can define $\nu^\bot(dv, dv')$ and $\tilde{\nu}^\bot(dv, dv')$ in a similar way. 
The acceptance probability has the form
\[
\alpha(v, v') = \min \left\{ 1, \frac{d \nu^\bot}{d \nu}(v, v') \right\}. 
\]
Letting 
\[
\beta(v, v') = \log\left( \frac{d \nu}{d \tilde{\nu}}(v, v') \right) = \log\left( \frac{d q(v, \cdot)}{d \tilde{q}(v, \cdot)}(v')\right) \log\left( \frac{d \mu_y}{d \mu_0}(v) \right),
\]
we have
\begin{equation}
\label{eq:c1_acc}
\frac{d \nu^\bot}{d \nu}(v, v') =  \frac{d \tilde{\nu}^\bot}{d \tilde{\nu}}(v, v') \exp\left( \beta(v', v) - \beta(v, v')\right).
\end{equation}
Recall equation \eqref{eq:tranpost}, 
\[
\frac{d \mu_y}{d \mu_0}(v) = \exp\left(-\potential(v ; y) - \frac12 \|v_{\rm ref}\|^2 - \langle v_{\rm ref} ,v \rangle \right),
\]
where $v_{\rm ref} = \prcov^{-\half}m_{\rm ref}$. 
The term $\frac12 \|v_{\rm ref}\|^2$ is cancelled out in $\beta(v', v)$ and $\beta(v, v')$, and hence we have
\begin{equation}
\label{eq:rho_2}
\beta(v, v') = -\potential(v ; y) - \langle v_{\rm ref} ,v \rangle + \log\left( \frac{d q(v, \cdot)}{d \tilde{q}(v, \cdot)}(v')\right).
\end{equation}
The first two terms in \eqref{eq:rho_2} are $\mu_0$-a.s.\ finite by the discussion preceding \eqref{eq:tranpost}.

To simplify the derivation of the right-hand side of \eqref{eq:c1_acc}, we apply the transformation \eqref{eq:trans_w}.
Then the pair of proposals $q(v, dv')$ and $\tilde{q}(v, dv')$ defined on the transformed parameters $(w, w')$ have the form
\begin{eqnarray*}
q(w, dw') & : & w' = \cD_\cA w + \cD_\cB \rand - \cD_\cG \gradientw,
\\
\tilde{q}(w, dw') & : & w' = \cD_\cA w + \cD_\cB \rand .
\end{eqnarray*}
The pair of measures $\tilde{\nu}^\bot(v, dv')$ and $\tilde{\nu}(v, dv')$ transformed onto $(w, w')$ have the form
\begin{eqnarray*}
\tilde{\nu}^\bot(dw, dw') & = & \tilde{q}(w, dw') \mu^w_0(dw) ,\\
\tilde{\nu}(dw, dw') & = & \tilde{q}(w', dw) \mu^w_0(dw') .
\end{eqnarray*}
where $\mu^w_0 = \normal(0, \cI)$ as the transformation \eqref{eq:trans_w} is unitary.
We want to derive the Radon-Nikodym derivatives
\[
\frac{d \tilde{\nu}^\bot}{d \tilde{\nu}}(w, w') = \frac{d \tilde{\nu}^\bot}{d \tilde{\nu}}(v, v'), 
\; {\rm and} \;
\frac{d q(w, \cdot)}{d \tilde{q}(w, \cdot)}(w') = \frac{d q(v, \cdot)}{d \tilde{q}(v, \cdot)}(v').
\]

Following the same derivation as in \ref{sec:proof:t1}, we can show that $\tilde{\nu}^\bot(dw, dw')$ and $\tilde{\nu}(dw, dw')$ follow the Gaussian measures
\[
\tilde{\nu}(dw, dw') = \normal(0, \cQ_1), \; {\rm and} \; \tilde{\nu}^\bot(dw, dw') = \normal(0, \cQ_2),
\]
where $\cQ_1$ and $\cQ_2$ are defined as in \eqref{eq:cov_operator}.
Thus we have
\begin{eqnarray*}
\log\left(\frac{d \tilde{\nu}^\bot}{d \tilde{\nu}}(w, w')\right) & = & 
- \frac12 \left\langle w', {\cD_\cB}^{-2}\left({\cD_\cA}^2 + {\cD_\cB}^2 - \cI\right) w' \right\rangle \\
&& + \frac12 \left\langle w, {\cD_\cB}^{-2}\left({\cD_\cA}^2 + {\cD_\cB}^2 - \cI\right) w \right\rangle.
\end{eqnarray*}
Both terms in the above equation are $\mu_0$-a.s.\ finite by condition (4) of Theorem \ref{theo:1}.  

Following Theorem 2.21 of \cite{DaPrato_1992}, we have 
\begin{eqnarray*}
\log\left( \frac{d q(w, \cdot)}{d \tilde{q}(w, \cdot)}(w') \right) & = & - \left \langle {\cD_\cB}^{-1} \cD_\cG  \gradientw, {\cD_\cB}^{-1} \left( w' - \cD_\cA w \right) \right \rangle  \\
&& - \frac{1}{2} \left \| {\cD_\cB}^{-1} \cD_\cG \gradientw\right \|^2 .
\end{eqnarray*}
Following the same argument as in \ref{sec:proof:t1}, we can show that
the Gaussian laws of $q(w, dw')$ and $\tilde{q}(w, dw')$ are
equivalent, and hence $\log( \frac{d q(w, \cdot)}{d \tilde{q}(w,
  \cdot)}(w') )$ is also guaranteed to be 
  $\mu_0$-a.s.\ finite.  This can be seen from the equation above, by the conditions of Theorem \ref{theo:1} 
  and Assumption \ref{assum1}(3), since the inner product of an element of 
  $\mathcal{H}$ with white noise is finite almost surely.

Transforming $\log \left (\frac{d \tilde{\nu}^\bot}{d \tilde{\nu}}(w, w') \right)$ back to $(v, v')$, we have
\begin{eqnarray*}
\log\left(\frac{d \tilde{\nu}^\bot}{d \tilde{\nu}}(v, v')\right) & = & 
- \frac12 \left\langle v', \cB^{-2}\left(\cA^2 + \cB^2 - \cI\right) v' \right\rangle \\
&& + \frac12 \left\langle v, \cB^{-2}\left(\cA^2 + \cB^2 - \cI\right) v \right\rangle.
\end{eqnarray*}
Transforming $\log \left ( \frac{d q(w, \cdot)}{d \tilde{q}(w, \cdot)}(w') \right )$
back to $(v, v')$ and substituting into \eqref{eq:rho_2}, we have
\[
\beta(v, v') = -\potential(v ; y) - \langle v_{\rm ref} ,v \rangle - \left \langle {\cB}^{-1} \cG  \gradientv, {\cB}^{-1} \left( v' - \cA v \right) \right \rangle - \frac{1}{2} \left \| {\cB}^{-1} \cG \gradientv\right \|^2.
\]
By collecting terms we can define
\begin{eqnarray*}
\rho(v, v') & = & - \potential(v; y) - \left\langle v_{\rm ref}, v \right\rangle - \frac12 \left\langle v, {\cB}^{-2}\left({\cA}^2 + {\cB}^2 - \cI\right) v  \right\rangle \\
&& - \left \langle {\cB}^{-1} \cG  \gradientv, {\cB}^{-1} \left( v' - \cA v \right) \right \rangle - \frac{1}{2} \left \| {\cB}^{-1} \cG \gradientv\right \|^2  ,
\end{eqnarray*}
and the acceptance probability has the form
\[
\alpha(v, v') = \min \left\{ 1, \exp \left ( \rho (v', v) - \rho(v, v')\right)\right\}.
\]
